\documentclass[12pt]{article}

\usepackage{graphicx, amssymb, latexsym, amsfonts, amsmath, lscape, amscd, mdframed,
amsthm, color, epsfig, mathrsfs, tikz, enumerate, esvect, algorithm, algpseudocode}
\usepackage[T1]{fontenc}
\usepackage[labelformat=simple]{subcaption}

\usepackage[textsize=footnotesize,color=green!40]{todonotes}
\usepackage{stmaryrd}
\usetikzlibrary{arrows}
\usetikzlibrary{decorations.pathreplacing,decorations.markings}
\newtheorem{theorem}{Theorem}[section]

\newtheorem{corollary}[theorem]{Corollary}

\newtheorem{lemma}[theorem]{Lemma}
\newtheorem{proposition}[theorem]{Proposition}
\newtheorem{question}[theorem]{Question}

\newtheorem{observation}[theorem]{Observation}

\theoremstyle{definition}
\newtheorem{defn}[theorem]{Definition}
\newtheorem{construction}[theorem]{Construction}
\newcommand\DELETE[1]{}

\tikzset{middlearrow/.style={
        decoration={markings,
            mark= at position 0.5 with {\arrow{#1}} ,
        },
        postaction={decorate}
    }
}

\begin{document}


\title{{\bf Monitoring arc-geodetic sets of oriented graphs}}
\author{
{\sc Tapas Das}$\,^{a}$, {\sc Florent Foucaud}$\,^{b}$,\\ {\sc Clara Marcille}$\,^{c}$, {\sc Pavan P D}$\,^{a,d}$, {\sc Sagnik Sen}$\,^{a}$ \\
\mbox{}\\
{\small $(a)$ Indian Institute of Technology Dharwad, Karnataka, India}\\
{\small $(b)$ Université Clermont Auvergne, CNRS, Clermont Auvergne INP, Mines Saint-\'Etienne,}\\ {\small LIMOS, 63000 Clermont-Ferrand, France}\\
{\small $(c)$ Univ. Bordeaux, CNRS,  Bordeaux INP, LaBRI, UMR 5800, F-33400, Talence, France}\\
{\small $(d)$ University of Turku, FI-20014, Finland}
}

\date{}

\maketitle

\begin{abstract}
Monitoring edge-geodetic sets in a graph are subsets of vertices such that every edge of the graph must lie on all the shortest paths between two vertices of the monitoring set. These objects were introduced in a work by Foucaud, Krishna and Ramasubramony Sulochana with relation to several prior notions in the area of network monitoring like distance edge-monitoring.

In this work, we explore the extension of those notions unto oriented graphs, modelling oriented networks, and call these objects monitoring arc-geodetic sets. We also define the lower and upper monitoring arc-geodetic number of an undirected graph as the minimum and maximum of the monitoring arc-geodetic number of all orientations of the graph. We determine the monitoring arc-geodetic number of fundamental graph classes such as bipartite graphs, trees, cycles, etc. Then, we characterize the graphs for which every monitoring arc-geodetic set is the entire set of vertices, and also characterize the solutions for tournaments. We also cover some complexity aspects by studying two algorithmic problems. We show that the problem of determining if an undirected graph has an orientation with the minimal monitoring arc-geodetic set being the entire set of vertices, is NP-hard. We also show that the problem of finding a monitoring arc-geodetic set of size at most $k$ is NP-complete when restricted to oriented graphs with maximum degree $4$.

\end{abstract}


\section{Introduction}
In the area of network monitoring, one may wish to detect faults in a network. We have in our hand probes that can detect the distances between each other through the 
standard procedure of \emph{ping}~\cite{bampas2015network, bejerano2003robust}. If we model the networks by finite, undirected simple graphs, whose vertices represent nodes and whose edges represent the connections between them, then the fault in the network is considered detected, if for example, the fault causes the distance between some pair of probes to increase. Hence, the goal is to choose a subset of vertices (representing the probes) of the simple graph, with the property that if any edge in the graph is removed (the connection fails), then it is detected by at least one pair of probes.

To solve such problems, Foucaud, Krishna and Ramasubramony Sulochana introduced the concept of monitoring edge-geodetic sets~\cite{foucaud2023monitoring,foucaud2023monitoring2} by taking inspiration from two areas: the concept of geodetic sets in graphs and its variants~\cite{harary1993geodetic}, and
the concept of distance edge-monitoring sets~\cite{foucaud2022monitoring}.
Let $G$ be a simple undirected graph. Two vertices $x, y$ \textit{monitor} an edge $e$ of $G$ if $e$ belongs to all shortest paths between $x$ and $y$. A set $M$ of vertices of $G$ is called a \textit{monitoring edge-geodetic set} (\textit{MEG-set} for short) of $G$ if, for every edge $e$ of $G$, there is a pair $x, y$ of vertices of $M$ that monitors $e$ (see Definition 1.1 in \cite{foucaud2023monitoring}). 
The size of a minimum MEG-set of $G$ is its \textit{monitoring edge-geodetic number}, denoted by $meg(G)$. Note that $V(G)$ is always an MEG-set of $G$, thus $meg(G)$ is always well-defined. 

The theory of monitoring edge-geodetic sets of graphs has been developed in several works including~\cite{foucaud2023monitoring,foucaud2024monitoring,foucaud2023monitoring2}. Some complexity aspects have been addressed in~\cite{bilo2024inapproximability,HASLEGRAVE202379} and monitoring edge-geodetic sets have been studied on networks in \cite{ma2024monitoring,tan2023monitoring,xu2024monitoring}.

In an attempt to broaden the scope of the application of monitoring edge-geodetic sets, we generalize the MEG-set problem to oriented graphs. Similar generalisations have been done in the past for geodetic sets~\cite{chartrand2000geodetic,lu2007geodetic,chang2004geodetic,kim2004geodetic,hung2006hull,hung2009hull,araujo2022hull,farrugia2005orientable} which are also defined using shortest paths between vertices, and our work draws inspiration from them. As a generalisation, we model the networks by finite, oriented graphs, whose vertices represent nodes and whose arcs represent the connections between them. To understand the significance of this generalisation, observe that while in undirected graphs, a shortest path from a vertex $u$ to a vertex $v$ is the same as a shortest path from $v$ to $u$, it is no longer the case in oriented graphs. That is, in an oriented graph, an arc $\overrightarrow{a}$ might lie on a shortest path from $u$ to $v$ but not on a shortest path from $v$ to $u$. Note that in an oriented graph, by a shortest path, we specifically mean a shortest \textit{directed} path, that is, an oriented path with all the arcs oriented in the same direction.

\begin{defn}
    In an oriented graph $\overrightarrow{G}$, two vertices $x$ and $y$ are said to \textit{monitor} an arc $\overrightarrow{a}$ if $\overrightarrow{a}$ belongs to the intersection of all shortest paths from $x$ to $y$ or the intersection of all shortest paths from $y$ to $x$. A subset $M$ of $V(\overrightarrow{G})$ is a \textit{monitoring arc-geodetic set} (\textit{MAG-set} for short) of $\overrightarrow{G}$ if for every arc $\overrightarrow{a} \in A(\overrightarrow{G})$, there exist two distinct vertices $u, v\in M$ such that $u$ and $v$ monitor $\overrightarrow{a}$. The \textit{monitoring arc-geodetic number}, denoted by $mag(\overrightarrow{G})$, is the minimum size of a MAG-set of $\overrightarrow{G}$.
\end{defn}
    Note that $V(\overrightarrow{G})$ is always an MAG-set of $\overrightarrow{G}$, thus $mag(\overrightarrow{G})$ is always well-defined. Also, by the definition, since an arc is monitored by a pair of distinct vertices, $mag(\overrightarrow{G}) \geq 2$ unless $\overrightarrow{G}$ is an oriented graph without any arcs, in which case $mag(\overrightarrow{G}) = 0$.

    While studying the geodetic sets for oriented graphs, Chartrand and Zhang~\cite{chartrand2000geodetic} defined the \textit{lower orientable geodetic number} $g^-(G)$ of $G$ as the minimum geodetic number of an orientation of $G$ and the \textit{upper orientable geodetic number} $g^+(G)$ as the maximum geodetic number of an orientation of $G$. Many works~\cite{lu2007geodetic,chang2004geodetic,dong2009upper,farrugia2005orientable,hung2006hull} have contributed to the study of geodetic sets in oriented graphs and lower and upper orientable geodetic numbers. For instance, Chartrand and Zhang~\cite{chartrand2000geodetic} and Chang, Tong and Wang~\cite{chang2004geodetic} determined upper and lower orientable geodetic numbers for some common graph classes like trees, paths, cycles, $r$-partite graphs and complete graphs. The characterization of oriented graphs with geodetic number equal to its order was given by Chartrand and Zhang~\cite{chartrand2000geodetic} and they also related the geodetic number with the diameter of the graph. Complexity aspects were studied by Ara{\'u}jo and Arraes~\cite{araujo2022hull} who showed that determining whether the geodetic number of an oriented graph $\overrightarrow{G}$ is at most $k$ is NP-hard, even if $\overrightarrow{G}$ has no directed cycle and its underlying graph is either bipartite or cobipartite. Our work will take inspiration from all these works due to the conceptual similarities between geodetic sets on oriented graphs and monitoring arc-geodetic sets. We define the following in the context of monitoring arc-geodetic sets.

\begin{defn}
    For an undirected graph $G$, the \textit{monitoring arc-geodetic spectrum} of $G$, denoted by $S(G)$, is the set $\{mag(\overrightarrow{G}) \colon\ \overrightarrow{G} \text{ is an orientation of } $G$\}$. The \textit{lower monitoring arc-geodetic number} of $G$ is denoted by $mag^-(G)$ and is equal to $\min (S(G))$. The \textit{upper monitoring arc-geodetic number} of $G$ is denoted by $mag^+(G)$ and is equal to $\max (S(G))$.
\end{defn}

An oriented graph $\overrightarrow{G}$ is said to be \textit{connected} if the underlying undirected graph is connected. A maximal connected subgraph of an oriented graph is called a \textit{component}. If $\overrightarrow{G}$ is an oriented graph with components $\{\overrightarrow{C_1}, \ldots, \overrightarrow{C_n}\}$, then finding an MAG-set of $\overrightarrow{G}$ is the same as finding MAG-sets of the individual components $\overrightarrow{C_i}$, since the concept of a shortest path between two distinct vertices only makes sense if they are in the same component. The union of the minimum MAG-sets of each component $\overrightarrow{C_i}$ gives us a minimum MAG-set of $\overrightarrow{G}$. Hence, in the forthcoming discussion, we will assume that all (oriented) graphs are connected unless otherwise mentioned. 

\subsection{Overview and organization of the paper}
In this work, we lay the groundwork for further exploration on MAG-sets and the monitoring arc-geodetic number.
\begin{itemize}
    \item Following the groundwork laid in this section by introducing essential concepts and giving key definitions, in Section~\ref{sec2}, we determine fundamental results which underpin our entire research work. We first pinpoint the vertices which are always a part of every MAG-set. Building upon this, we give detailed explanations on finding MAG-sets of graphs by determining them for some fundamental graph classes such as trees, paths and cycles. We also study how to determine the upper and lower monitoring arc-geodetic numbers for graph classes, by carrying out an analysis and determining the precise values of $mag^-$ and $mag^+$ for each class studied above.
    
    \item In Section~\ref{sec3}, we try to learn more about the behaviour of MAG-sets in a graph by focusing on those oriented graphs for which the minimimal MAG-set is always the entire set of vertices. We call such graphs MAG-extremal, drawing an analogy with the MEG-extremal graphs defined with respect to monitoring edge-geodetic sets~\cite{foucaud2023monitoring}. We prove a characterization result for MAG-extremal graphs which identifies the properties of a vertex which is always in every MAG-set of the graph. We also illustrate via examples, the lack of a straightforward relation between $meg$ and $mag^+$ of an undirected graph.
    
    \item In Section~\ref{sec4}, we deal with MAG-sets within tournaments, another fundamental graph class. We completely characterize the MAG-sets of tournaments by proving that the monitoring arc-geodetic number of a tournament of order $n$ is at least $n - 1$.
    We demonstrate that this bound is tight by proving the existence of tournaments which achieve the lower and upper bounds (trivially $n$). This finding contrasts with the behavior of monitoring edge-geodetic sets of complete graphs, which always include all the vertices~\cite{foucaud2023monitoring}. It is also distinct from the geodetic number of tournaments of order $n$, which can take any value in the set $\{2, \ldots, n\}$~\cite{chartrand2000geodetic}.
    
    \item In Section~\ref{sec5}, we focus on complexity aspects of some problems related to MAG-sets. The study of MAG-extremal graphs inspires the \textsc{MAG$^+$-set} problem, which asks whether for an undirected graph $G$, $mag^+(G) = |V(G)|$. We show that this problem is NP-hard. Additionally, we study the \textsc{$k$-MAG-set} problem, which asks whether there exists an MAG-set of size $k$ in an oriented graph $\overrightarrow{G}$. Our results prove that this problem is NP-complete even when restricted to acyclic oriented graphs with maximum degree 4.
    
    \item Finally, we conclude in Section~\ref{sec6} by proposing several open questions and potential research directions, paving the way for future exploration in this area.
\end{itemize}

\section{Preliminary results}
\label{sec2}
We begin with some results which are analogous to the known results about the monitoring edge-geodetic number~\cite{foucaud2023monitoring,foucaud2024monitoring}. Let $\overrightarrow{G}$ be an oriented graph. We say that a vertex $u$ is an \textit{in-neighbor} (resp., \textit{out-neighbor}) of a vertex $v$ in $\overrightarrow{G}$ if $\overrightarrow{uv}$ (resp., $\overrightarrow{vu}$) is an arc in $\overrightarrow{G}$. The set of in-neighbors (resp., out-neighbors) of $u$ in $\overrightarrow{G}$ is denoted by $N^-(u)$ (resp., $N^+(u)$). A vertex $u \in V(\overrightarrow{G})$ is a source (resp., sink) if $N^-(u) = \emptyset$ (resp., $N^+(u) = \emptyset$).

\subsection{Useful propositions}

\begin{proposition}\label{prop:SnS}
Let $\overrightarrow{G}$ be an oriented graph and $u$ be a source or a sink of $\overrightarrow{G}$. Then $u$ is in every MAG-set of $\overrightarrow{G}$.
\end{proposition}
\begin{proof}
    Let $u$ be a source (resp., sink) and let $\overrightarrow{uv}$ (resp., $\overrightarrow{vu}$) be an arc in $\overrightarrow{G}$. For the arc $\overrightarrow{uv}$ (resp., $\overrightarrow{vu}$) to be monitored, $u$ must belong to every MAG-set since there is no shortest path between two distinct vertices of $\overrightarrow{G}$ with $u$ as an intermediate vertex.
\end{proof}

We can give the following simple result for bipartite graphs as a consequence of Proposition~\ref{prop:SnS}.

\begin{proposition}
\label{prop:bipartite}
    Let $G$ be a bipartite graph. Then $mag^+(G) = |V(G)|$.
\end{proposition}
\begin{proof}
    Let $G$ be a bipartite graph with parts $A$ and $B$. We can always orient the edges of $G$ from $A$ to $B$ in such a way that the vertices in the part $A$ are sources and the vertices in the part $B$ are sinks. By Proposition~\ref{prop:SnS}, the unique minimum MAG-set of this orientation is the set of all vertices.
\end{proof}

Two distinct vertices $u$ and $v$ in an oriented graph $\overrightarrow{G}$ are \textit{twins} if $N^+(u) = N^+(v)$ and $N^-(u) = N^-(v)$.

\begin{proposition}
    In an oriented graph $\overrightarrow{G}$, every pair of twins is in every MAG-set of $\overrightarrow{G}$.
\end{proposition}
\begin{proof}
    Let $u, v$ be twins in $\overrightarrow{G}$. If $u$ and $v$ are sources or sinks, then the result is true by Proposition~\ref{prop:SnS}. Hence, let $u$ and $v$ be neither sources nor sinks and assume without loss of generality that $u$ is not a part of an MAG-set $M$ of $\overrightarrow{G}$. Let $u'$ be an out-neighbor of $u$ and let the arc $\overrightarrow{uu'}$ be monitored by the vertices $a$ and $b$ in $M$. This means that there exists a shortest path $\overrightarrow{P} = a \ldots x u u' \ldots b$ between $a$ and $b$ (where $u'$ can be the same as $b$). Consider the path $\overrightarrow{P'} = a \ldots x v u' \ldots b$, where the subpaths between $a$ and $x$, and $u'$ and $b$ are the same in $\overrightarrow{P}$ and $\overrightarrow{P'}$. Clearly $\overrightarrow{P}$ and $\overrightarrow{P'}$ are of the same length, but $\overrightarrow{uu'}$ is not an arc in $\overrightarrow{P'}$, contradicting the fact that $\overrightarrow{uu'}$ is monitored by $a$ and $b$. Hence, $u$ must be a part of every MAG-set of $\overrightarrow{G}$.
\end{proof}

\subsection{Some basic oriented graph classes}

Orientations of trees have the property that there exists a unique shortest path between any two vertices in the graph. Hence, the calculation of the monitoring arc-geodetic numbers of these graphs is straightforward. 

\begin{lemma}
\label{lemma:MAGtrees}
Let $\overrightarrow{G}$ be an orientation of a tree. Then there is a unique minimum MAG-set of $\overrightarrow{G}$, and it is the set of its sources and sinks.
\end{lemma}
\begin{proof}
    Consider the set of sources and sinks of $\overrightarrow{G}$. Let $\overrightarrow{a}$ be an arc of $\overrightarrow{G}$. The arc $\overrightarrow{a}$ lies on the unique path between some source $u$ and some sink $v$, and this path is the only shortest path between $u$ and $v$. Hence $\overrightarrow{a}$ is monitored, and therefore, the set of sources and sinks is an MAG-set of $\overrightarrow{G}$. Proposition~\ref{prop:SnS} proves the lower bound on the size of a minimum MAG-set of $\overrightarrow{G}$, and completes the proof.
\end{proof}

For trees, the values of $mag^-$ and $mag^+$ can hence be easily determined by minimising and maximising the number of sources and sinks in the graph, respectively. For trees, it is possible to give an orientation such that the minimum MAG-set of the graph will consist only of the vertices of degree 1 in the tree.

\begin{corollary}
\label{cor:mag+tree}
Let $G$ be a tree, then $mag^+(G) = |V(G)|$ and $mag^-(G)$ is equal to the number of vertices of degree 1.
\end{corollary}
\begin{proof}
    Since $G$ is a tree, it is bipartite, and by Proposition~\ref{prop:bipartite}, we have $mag^+(G) = |V(G)|$.

    Let $u$ be a vertex of degree 1 in $G$. Choose $u$ as the root of $G$ and orient all the arcs away from $u$. In this orientation, between any two degree 1 vertices (leaves), there is a unique shortest path, and hence, we get our result using Lemma~\ref{lemma:MAGtrees}.
\end{proof}

We define an undirected path $P_n$ of length $n$ to be a graph with $V(P_n) = \{v_1, \ldots, v_n\}$ and $E(P_n) = \{v_1v_2, \ldots, v_{n - 1}v_n\}$. Since paths are trees, the monitoring arc-geodetic number of paths depends only on the number of sources or sinks in the paths. Hence, we have the following immediate corollary. 

\begin{corollary}
For $n \geq 1$, $mag^-(P_n) = 2$ and $mag^+(P_n) = n$.
\end{corollary}

The cycle of length $n$, where $n \geq 3$, is the graph $C_n$ with $V(C_n) = \{1, 2, \ldots, n\}$ and $E(C_n) = \{v_1v_2, \ldots, v_{n - 1}v_n\} \cup \{v_nv_1\}$. Unlike in trees, between any two distinct vertices in a cycle, there are two paths. Hence, an MAG-set of a cycle does not necessarily consist of only sources and sinks. In the next result, we characterize the monitoring arc-geodetic number of cycles. For this purpose, we classify all the possible orientations of a cycle $C_n$ into $\overrightarrow{C_n^0}, \overrightarrow{C_n^1}, \overrightarrow{C_n^2}$ and $\overrightarrow{C_n^3}$, and describe them as follows:

\sloppy
\begin{itemize}
    \item $\overrightarrow{C_n^0}$ is an oriented cycle with no sources and sinks. That is, $A(\overrightarrow{C_n^0}) = \{\overrightarrow{v_1v_2}, \ldots, \overrightarrow{v_{n - 1}v_n}\} \cup \{\overrightarrow{v_n v_1}\}$.
    \item $\overrightarrow{C_n^1}$ is an oriented even cycle with exactly one source and one sink which are at distance $\frac{n}{2}$ from each other. That is, $A(\overrightarrow{C_n^1}) = \{\overrightarrow{v_1v_2}, \ldots, \overrightarrow{v_{\frac{n}{2}}v_{\frac{n}{2} + 1}}\} \cup \{\overrightarrow{v_{\frac{n}{2} + 2}v_{\frac{n}{2} + 1}}, \ldots, \overrightarrow{v_n,v_{n - 1}}\} \cup \{\overrightarrow{v_1 v_n}\}$.
    \item $\overrightarrow{C_n^2}$ is an oriented cycle with exactly one source and one sink which are at a distance $d$ such that $d \neq \frac{n}{2}$. That is, $A(\overrightarrow{C_n^2}) = \{\overrightarrow{v_1v_2}, \ldots, \overrightarrow{v_{d}v_{d + 1}}\} \cup \{\overrightarrow{v_{d + 2}v_{d + 1}}, \ldots, \overrightarrow{v_n,v_{n - 1}}\} \cup \{\overrightarrow{v_1 v_n}\}$, where $d \neq \frac{n}{2}$.
    \item $\overrightarrow{C_n^3}$ is an oriented cycle with more than one source and more than one sink.
\end{itemize}

Figure~\ref{fig:orexcep} gives an illustration of $\overrightarrow{C_n^0}, \overrightarrow{C_n^1}$ and $\overrightarrow{C_n^2}$. Observe that $\overrightarrow{C_n^0}, \overrightarrow{C_n^1}, \overrightarrow{C_n^2}$ and $\overrightarrow{C_n^3}$ describe all possible orientations of the cycle $C_n$.

\begin{figure}
    \centering
        \begin{subfigure}{0.3\textwidth}
    \centering
    \begin{tikzpicture}[scale=0.5]
        \fill[black] (0,0) (0:3cm) circle[radius=3pt];
        \fill[black] (0,0) (15:3cm) circle[radius=3pt];
        \fill[black] (0,0) (30:3cm) circle[radius=3pt];
        \fill[black] (0,0) (60:3cm) circle[radius=3pt];
        \fill[black] (0,0) (90:3cm) circle[radius=3pt];
        \fill[black] (0,0) (120:3cm) circle[radius=3pt];
        \fill[black] (0,0) (150:3cm) circle[radius=3pt];
        \fill[black] (0,0) (165:3cm) circle[radius=3pt];
        \fill[black] (0,0) (180:3cm) circle[radius=3pt];
        \fill[black] (0,0) (195:3cm) circle[radius=3pt];
        \fill[black] (0,0) (210:3cm) circle[radius=3pt];
        \fill[black] (0,0) (240:3cm) circle[radius=3pt];
        \fill[black] (0,0) (270:3cm) circle[radius=3pt];
        \fill[black] (0,0) (300:3cm) circle[radius=3pt];
        \fill[black] (0,0) (330:3cm) circle[radius=3pt];
        \fill[black] (0,0) (345:3cm) circle[radius=3pt];
        \node at (90:3.6) {$v_1$};
        \node at (60:3.6) {$v_2$};
        \node at (120:3.6) {$v_{n}$};
        \node at (-90:3.6) {$v_{k + 1}$};
        \node at (-60:3.6) {$v_k$};
        \node at (-120:3.6) {$v_{k + 2}$};
        \draw[very thick,black,middlearrow={>}] ([shift=(60:3cm)]0,0) arc (60:30:3cm);
        \draw[very thick,black,middlearrow={>}] ([shift=(90:3cm)]0,0) arc (90:60:3cm);
        \draw[very thick,black,middlearrow={>}] ([shift=(120:3cm)]0,0) arc (120:90:3cm);
        \draw[very thick,black,middlearrow={>}] ([shift=(150:3cm)]0,0) arc (150:120:3cm);
        \draw[very thick,black,middlearrow={>}] ([shift=(270:3cm)]0,0) arc (270:240:3cm);
        \draw[very thick,black,middlearrow={>}] ([shift=(240:3cm)]0,0) arc (240:210:3cm);
        \draw[very thick,black,middlearrow={>}] ([shift=(300:3cm)]0,0) arc (300:270:3cm);
        \draw[very thick,black,middlearrow={>}] ([shift=(330:3cm)]0,0) arc (330:300:3cm);
    \end{tikzpicture}
    \caption{$\overrightarrow{C_n^{0}}$}
    \label{fig:first}
    \end{subfigure}
    \begin{subfigure}{0.3\textwidth}
    \centering
    \begin{tikzpicture}[scale=0.5]
        \fill[black] (0,0) (0:3cm) circle[radius=3pt];
        \fill[black] (0,0) (15:3cm) circle[radius=3pt];
        \fill[black] (0,0) (30:3cm) circle[radius=3pt];
        \fill[black] (0,0) (60:3cm) circle[radius=3pt];
        \fill[black] (0,0) (90:3cm) circle[radius=3pt];
        \fill[black] (0,0) (120:3cm) circle[radius=3pt];
        \fill[black] (0,0) (150:3cm) circle[radius=3pt];
        \fill[black] (0,0) (165:3cm) circle[radius=3pt];
        \fill[black] (0,0) (180:3cm) circle[radius=3pt];
        \fill[black] (0,0) (195:3cm) circle[radius=3pt];
        \fill[black] (0,0) (210:3cm) circle[radius=3pt];
        \fill[black] (0,0) (240:3cm) circle[radius=3pt];
        \fill[black] (0,0) (270:3cm) circle[radius=3pt];
        \fill[black] (0,0) (300:3cm) circle[radius=3pt];
        \fill[black] (0,0) (330:3cm) circle[radius=3pt];
        \fill[black] (0,0) (345:3cm) circle[radius=3pt];
        \node at (90:3.6) {$v_1$};
        \node at (60:3.6) {$v_2$};
        \node at (120:3.6) {$v_{n}$};
        \node at (-90:3.6) {$v_{\frac{n}{2} + 1}$};
        \node at (-60:3.6) {$v_{\frac{n}{2}}$};
        \node at (-120:3.6) {$v_{\frac{n}{2} + 2}$};
        \draw[very thick,black,middlearrow={>}] ([shift=(60:3cm)]0,0) arc (60:30:3cm);
        \draw[very thick,black,middlearrow={>}] ([shift=(90:3cm)]0,0) arc (90:60:3cm);
        \draw[very thick,black,middlearrow={>}] ([shift=(90:3cm)]0,0) arc (90:120:3cm);
        \draw[very thick,black,middlearrow={>}] ([shift=(120:3cm)]0,0) arc (120:150:3cm);
        \draw[very thick,black,middlearrow={>}] ([shift=(240:3cm)]0,0) arc (240:270:3cm);
        \draw[very thick,black,middlearrow={>}] ([shift=(210:3cm)]0,0) arc (210:240:3cm);
        \draw[very thick,black,middlearrow={>}] ([shift=(300:3cm)]0,0) arc (300:270:3cm);
        \draw[very thick,black,middlearrow={>}] ([shift=(330:3cm)]0,0) arc (330:300:3cm);
    \end{tikzpicture}
    \caption{$\overrightarrow{C_n^{1}}$}
    \label{fig:second}
    \end{subfigure}
    \begin{subfigure}{0.3\textwidth}
    \centering
    \begin{tikzpicture}[scale=0.5]
        \fill[black] (0,0) (0:3cm) circle[radius=3pt];
        \fill[black] (0,0) (15:3cm) circle[radius=3pt];
        \fill[black] (0,0) (30:3cm) circle[radius=3pt];
        \fill[black] (0,0) (60:3cm) circle[radius=3pt];
        \fill[black] (0,0) (90:3cm) circle[radius=3pt];
        \fill[black] (0,0) (120:3cm) circle[radius=3pt];
        \fill[black] (0,0) (150:3cm) circle[radius=3pt];
        \fill[black] (0,0) (165:3cm) circle[radius=3pt];
        \fill[black] (0,0) (180:3cm) circle[radius=3pt];
        \fill[black] (0,0) (195:3cm) circle[radius=3pt];
        \fill[black] (0,0) (210:3cm) circle[radius=3pt];
        \fill[black] (0,0) (240:3cm) circle[radius=3pt];
        \fill[black] (0,0) (270:3cm) circle[radius=3pt];
        \fill[black] (0,0) (300:3cm) circle[radius=3pt];
        \fill[black] (0,0) (330:3cm) circle[radius=3pt];
        \fill[black] (0,0) (345:3cm) circle[radius=3pt];
        \node at (90:3.6) {$v_1$};
        \node at (60:3.6) {$v_2$};
        \node at (120:3.6) {$v_{n}$};
        \node at (-90:3.6) {$v_{d + 2}$};
        \node at (-60:3.6) {$v_{d + 1}$};
        \node at (-120:3.6) {$v_{d + 3}$};
        \draw[very thick,black,middlearrow={>}] ([shift=(60:3cm)]0,0) arc (60:30:3cm);
        \draw[very thick,black,middlearrow={>}] ([shift=(90:3cm)]0,0) arc (90:60:3cm);
        \draw[very thick,black,middlearrow={>}] ([shift=(90:3cm)]0,0) arc (90:120:3cm);
        \draw[very thick,black,middlearrow={>}] ([shift=(120:3cm)]0,0) arc (120:150:3cm);
        \draw[very thick,black,middlearrow={>}] ([shift=(240:3cm)]0,0) arc (240:270:3cm);
        \draw[very thick,black,middlearrow={>}] ([shift=(210:3cm)]0,0) arc (210:240:3cm);
        \draw[very thick,black,middlearrow={>}] ([shift=(270:3cm)]0,0) arc (270:300:3cm);
        \draw[very thick,black,middlearrow={>}] ([shift=(330:3cm)]0,0) arc (330:300:3cm);
    \end{tikzpicture}
    \caption{$\overrightarrow{C_n^{2}}$}
    \label{fig:third}
    \end{subfigure}
    \caption{The orientations $\overrightarrow{C_n^0}, \overrightarrow{C_n^1}$ and $\overrightarrow{C_n^2}$ of $C_n$.}
    \label{fig:orexcep}
\end{figure}

\begin{proposition}
\label{prop:magcycle} Let $\overrightarrow{C_n^0}$, $\overrightarrow{C_n^1}$, $\overrightarrow{C_n^2}$ and $\overrightarrow{C_n^3}$ be orientations of $C_n$ as described above. Then,
    \begin{enumerate}[(i)]
        \item $mag(\overrightarrow{C_n^0}) = 2$.
        \item $mag(\overrightarrow{C_n^1}) = 4$.
        \item $mag(\overrightarrow{C_n^2}) = 3$.
        \item $mag(\overrightarrow{C_n^3})$ is equal to the number of sources and sinks in $\overrightarrow{C_n^3}$.
    \end{enumerate}
\end{proposition}
\begin{proof}
    \begin{enumerate}[(i)]
        \item As $\overrightarrow{C_n^0}$ has no sources and sinks, we can claim that $\{v_i, v_j\}$ is a minimal MAG-set of $\overrightarrow{C_n^0}$, where $v_i$ and $v_j$ are arbitrary distinct vertices of $\overrightarrow{C_n^0}$. This is straightforward to see as the arcs $\{\overrightarrow{v_i v_{i + 1}}, \ldots, \overrightarrow{v_{j - 1} v_j}\}$ are monitored by the unique shortest path from $v_i$ to $v_j$ and the remaining arcs $\{\overrightarrow{v_j v_{j + 1}}, \ldots, \overrightarrow{v_{i - 1} v_i}\}$ of $\overrightarrow{C_n^0}$ are monitored by the unique shortest path from $v_j$ to $v_i$. Hence, $mag(\overrightarrow{C_n^0}) = 2$.
        \item As $\overrightarrow{C_n^1}$ has a source $v_1$ and a sink $v_{\frac{n}{2}}$, by Proposition~\ref{prop:SnS}, $v_1$ and $v_{\frac{n}{2}}$ are part of every MAG-set. However, there are two shortest paths between these two vertices, and hence, this pair of vertices does not monitor any arc in $\overrightarrow{C_n^1}$. Without loss of generality, we can first choose an arbitrary vertex $v_i, 1 < i < \frac{n}{2}$ to be a part of the MAG-set. Now, $v_1$ and $v_i$ monitor all arcs $\{\overrightarrow{v_1 v_2}, \ldots, \overrightarrow{v_{i - 1}v_i}\}$, and $v_i$ and $v_{\frac{n}{2}}$ monitor all arcs $\{\overrightarrow{v_i v_{i + 1}}, \ldots, \overrightarrow{v_{\frac{n}{2} - 1} v_{\frac{n}{2}}}\}$. Clearly, these three vertices are not enough to monitor all arcs of $\overrightarrow{C_n^1}$. Hence, we must choose an arbitrary vertex $v_j, \frac{n}{2} < j \leq n$ to also be a part of the MAG-set. By a similar argument as above, we can see that these four vertices monitor all arcs of $\overrightarrow{C_n^1}$. Hence, $mag(\overrightarrow{C_n^1}) = 4$.
        
        \item As $\overrightarrow{C_n^2}$ has a source $v_1$ and a sink $v_{d+1}$, by Proposition~\ref{prop:SnS}, $v_1$ and $v_{d + 1}$ are part of every MAG-set. Also $v_1 \ldots v_{d + 1}$ is the unique shortest path from $v_1$ to $v_{d + 1}$ and hence all arcs in this path are monitored by these two vertices. To monitor the remaining arcs of the cycle, we choose an arbitrary vertex $v_i, d + 1 < i \leq n$ to be a part of the MAG-set. Clearly, $v_1$ and $v_i$ monitor the arcs $\{\overrightarrow{v_1 v_n}\} \cup \{v_n v_{n - 1}, \ldots, \overrightarrow{v_{i + 1}v_i}\}$, and $v_i$ and $v_{d + 1}$ monitor the arcs $\{\overrightarrow{v_i v_{i - 1}}, \ldots, \overrightarrow{v_{d + 2} v_{d + 1}}\}$. Hence, $mag(\overrightarrow{C_n^2}) = 3$.
        
        \item As $\overrightarrow{C_n^3}$ has strictly more than one source and more than one sink, every arc lies on a unique shortest path between a source and a sink. Hence, the sources and sinks of $\overrightarrow{C_n^3}$ are necessary and sufficient to monitor all arcs of the graph.
    \end{enumerate}\end{proof}

The following corollary is a  direct observation from Proposition~\ref{prop:magcycle}
\begin{corollary}
\label{cor:cycle}
    For $n \geq 3$, $mag^-(C_n) = 2$ and $mag^+(C_3) = 3$. For $n \geq 4$, if $n$ is odd, then $mag^+(C_n) = n - 1$, and if $n$ is even, then $mag^+(C_n) = n$. 
\end{corollary}
\begin{proof}
    If the cycle $C_n$ is oriented as $\overrightarrow{C_n^0}$, then $mag(\overrightarrow{C_n^0}) = 2$ implies $mag^-(C_n) = 2$. To obtain the orientation with maximum possible monitoring arc-geodetic number, we maximise the number of sources and sinks in the orientation of $C_n$. This maximum number is $n - 1$ for odd cycles and $n$ for even cycles. The cycle $C_3$ is an exception since the transitive triangle, which has the maximum number of sources and sinks, has monitoring arc-geodetic number equal to 3.
\end{proof}

This result can be used to get a trivial bound for $mag^+(G)$ in terms of its \emph{girth}, that is, in terms of the length of a shortest cycle in $G$.

\begin{corollary}
    Let $G$ be a graph with girth $g \geq 3$. If $g$ is odd, then, $mag^+(G) \geq g - 1$ and if $g$ is even, then $mag^+(G) \geq g$.
\end{corollary}
\begin{proof}
    An orientation $\overrightarrow{G}$ of $G$ with $mag(\overrightarrow{G}) \geq g - 1$ can be obtained by selecting an arbitrary induced cycle $C$ of length $g$, then orienting the arcs such that the vertices of this cycle are alternating sources and sinks (except at most one vertex if the cycle is of odd length) in $\overrightarrow{G}$. This is always possible, for if not, then there exist two vertices of the cycle which are non-adjacent in the induced cycle but are adjacent in $G$. This however, contradicts the girth of the graph $G$. Thus, the number of sources and sinks in $\overrightarrow{G}$ by Corollary~\ref{cor:cycle} is at least $g - 1$ if $g$ is odd, and is at least $g$ if $g$ is even. Proposition~\ref{prop:SnS} completes the proof.
\end{proof}

\section{MAG-extremal oriented graphs}
\label{sec3}
In this section, we characterize the graphs for which the minimum MAG-set is the entire set of vertices, that is, $mag(\overrightarrow{G}) = n$. We shall call such graphs \textit{MAG-extremal graphs}, by taking inspiration from the \textit{MEG-extremal} graphs which were first defined in~\cite{foucaud2024monitoring}. They characterized all MEG-extremal graphs by the following result.

\begin{theorem}[\cite{foucaud2024monitoring}, Corollary 4.2]
    Let $G$ be a graph of order $n$. Then, $meg(G) = n$ if and only if for every $v \in V(G)$, there exists $u \in N(v)$ such that any induced $2$-path $uvx$ is part of a $4$-cycle.
\end{theorem}

The analogous characterization for MAG-extremal graphs is as follows.

\begin{theorem}\label{theorem:mageqn}
    An oriented graph $\overrightarrow{G}$ is an MAG-extremal graph if and only if every vertex $v \in V(\overrightarrow{G})$ satisfies one of the following conditions:
    \begin{enumerate}[(i)]
        \item $v$ is either a source or a sink,
        \item there exists a vertex $u \in N^-(v)$ such that for every vertex $w \in N^+(v)$, there is a directed path of length at most 2 from $u$ to $w$ not visiting $v$,
        \item there exists a vertex $w \in N^+(v)$ such that for every vertex $u \in N^-(v)$, there is a directed path of length at most 2 from $u$ to $w$ not visiting $v$.
    \end{enumerate}
\end{theorem}
\begin{proof}
    To prove necessity, let us assume that for every vertex $v$, (i) (ii) or (iii) holds, but that, by contradiction, $mag(\overrightarrow{G}) < n$. Hence, there exists a vertex $v \in V(\overrightarrow{G})$ which is not in some MAG-set $M$. By Proposition~\ref{prop:SnS}, $v$ cannot be a source or sink. Suppose that there exists a vertex $u \in N^-(v)$ that satisfies condition (ii) of the hypothesis. Let the arc $\overrightarrow{uv} \in A(\overrightarrow{G})$ be monitored by some distinct vertices $a, b$ in $M$ with $\overrightarrow{P} = a, \ldots u, v, w, \ldots, b$ being a shortest path from $a$ to $b$. But by condition (ii), there exists another path from $a$ to $b$ which is of equal or shorter (by one unit) length than $\overrightarrow{P}$, which does not pass through $v$, creating a contradiction to the fact that $\overrightarrow{uv}$ is monitored by $a$ and $b$. Similarly, if we suppose that there exists a vertex $w \in N^+(v)$ that satisfies condition (iii) of the hypothesis, then a contradiction arises along similar lines.

    We prove sufficiency by proving its contrapositive. Suppose there exists a vertex $v$ of $\overrightarrow{G}$ that does not satisfy any of (i) (ii), nor (iii). Thus, $v$ is neither a source nor a sink, so, $v$ has both incoming and outgoing arcs. Moreover, for every vertex $u \in N^-(v)$, there exists a vertex $w \in N^+(v)$, such that $u$ and $w$ are non-adjacent and every directed path of length 2 from $u$ to $w$ passes through $v$. We claim that $M = V(\overrightarrow{G}) \setminus \{v\}$ is an MAG-set of $\overrightarrow{G}$. To see this, observe that every arc $\overrightarrow{uv} \in A(\overrightarrow{G})$, is monitored by $u$ and $w$, since $uvw$ is a shortest path from $u$ to $w$, with our assumption ensuring that any other shortest path from $u$ to $w$ passes through $v$. Similarly, the assumption that for every vertex $w \in N^+(v)$, there exists a vertex $u \in N^-(v)$, such that $u$ and $w$ are non-adjacent and every directed path of length 2 from $u$ to $w$ passes through $v$, ensures that all arcs $\overrightarrow{vw} \in A(\overrightarrow{G})$ are monitored by $u$ and $w$. Hence, $M$ is an MAG-set of $\overrightarrow{G}$, that is, $mag(\overrightarrow{G}) < n$, as required. 
\end{proof}

We have seen in Proposition~\ref{prop:bipartite} that bipartite graphs can be oriented as MAG-extremal graphs. Some other examples of graph families that can be oriented to be MAG-extremal are comparability graphs or transitively orientable graphs, and complete graphs. We elaborate on the monitoring arc geodetic number of complete graphs in Section~\ref{sec:tournaments}. On the other hand, we can conclude from Theorem~\ref{theorem:mageqn} that 2-connected non-bipartite graphs of girth at least 5 cannot be oriented to be MAG-extremal. It would be interesting to find out more examples of families that can or cannot be oriented to be MAG-extremal.

\subsection{Relation between $meg$ and $mag^+$ of an undirected graph}

Dev, Dey, Foucaud, Krishna and Ramasubramony Sulochana~\cite{foucaud2023monitoring} gave the following result for the monitoring edge-geodetic number of cycles.

\begin{theorem}[{\cite[Theorem 3.4]{foucaud2023monitoring}}]
    Given an $n$-cycle graph $C_n$, for $n = 3$ and $n \geq 5$, $meg(C_n) = 3$. Moreover, $meg(C_4) = 4$.
\end{theorem}

Proposition~\ref{prop:magcycle} shows us that for the family of cycles $C_n$, we can have for some orientations $\overrightarrow{C_n}$, $mag(\overrightarrow{C_n}) < meg(C_n)$. Thus, it is natural to ask if  there also exist graphs $G$ such that for any orientation $\overrightarrow{G}$ of $G$, we have $mag(\overrightarrow{G}) > meg(G)$. We answer this question positively through the following 
construction. 

\begin{construction}
\label{cons:mag<mag-}
    Let us  start by taking a path $x'xyy'$ on four
    vertices. Then, take $j$ vertices
    $z_1, z_2, \ldots, z_j$ and make each of them adjacent to $x$ and $y$, 
    for $j \geq 1$. 
    Furthermore, for each $i$, make $z_i$ adjacent to a new vertex $z'_i$. The so-obtained graph is called $G_j$. This graph is illustrated in Figure~\ref{fig:mag<mag-}.
\end{construction}

\begin{figure}
    \centering
    \begin{tikzpicture}[inner sep=0.7mm, scale = 0.6]
    \node[draw, circle, line width=1pt, fill=black] (y') at (-2, 1) [label=below: $y'$]{};
    \node[draw, circle, line width=1pt, fill=black] (y) at (0, 1) [label=below: $y$]{};
    \node[draw, circle, line width=1pt, fill=black] (x) at (0, 11) [label=above: $x$]{};
    \node[draw, circle, line width=1pt, fill=black] (x') at (-2, 11) [label=below: $x'$]{};
    \node[draw, circle, line width=1pt, fill=black] (z1) at (3, 6) [label=below: $z_1$]{};
    \node[draw, circle, line width=1pt, fill=black] (z1') at (5, 6) [label=below: $z_1'$]{};
    \node[draw, circle, line width=1pt, fill=black] (z2) at (7, 6) [label=below: $z_2$]{};
    \node[draw, circle, line width=1pt, fill=black] (z2') at (9, 6) [label=below: $z_2'$]{};
    \node[draw, circle, line width=1pt, fill=black] (zj) at (13, 6) [label=below: $z_j$]{};
    \node[draw, circle, line width=1pt, fill=black] (zj') at (15, 6) [label=below: $z_j'$]{};
    \node[] (xn1) at (10, 6){\Huge.};
    \node[] (xn1) at (11, 6){\Huge.};
    \node[] (xn1) at (12, 6){\Huge.};
    \draw[black] (x) -- (x');
    \draw[black] (x) -- (y);
    \draw[black] (y) -- (y');
    \draw[black] (y) -- (z1);
    \draw[black] (z1) -- (z1');
    \draw[black] (x) -- (z2);
    \draw[black] (y) -- (z2);
    \draw[black] (z2) -- (z2');
    \draw[black] (x) -- (z1);
    \draw[black] (y) -- (zj);
    \draw[black] (zj) -- (zj');
    \draw[black] (x) -- (zj);
    \end{tikzpicture}
    \caption{The graph $G_j$.}
    \label{fig:mag<mag-}
\end{figure}

To find $meg(G_j)$ we make use of Lemma 2.1 from \cite{foucaud2023monitoring}.

\begin{lemma}[{\cite[Lemma 2.1]{foucaud2023monitoring}}]
\label{lem:megsimp}
    In a graph $G$ with at least one edge, any vertex of degree 1 belongs to any edge-geodetic set and thus, to any MEG-set of $G$.
\end{lemma}

\begin{observation}
    Let $j \geq 1$, be an integer and let $G_j$ be as per Construction~\ref{cons:mag<mag-}. Then $meg(G_j) = j + 2$.
\end{observation}
\begin{proof}
    The set of all vertices of degree $1$ is an MEG-set of $G_j$. 
\end{proof}

Recall that for an undirected graph $G$, the monitoring arc-geodetic spectrum $S(G)$ of $G$, is the set $\{mag(\overrightarrow{G}) \colon\ \overrightarrow{G} \text{ is an orientation of } $G$\}$. We next show that the graph $G_j$ constructed using Construction~\ref{cons:mag<mag-} has the property $meg(G_j) \notin S(G_j)$, that is, $meg^{-}(G_j) \geq j+3$.

\begin{figure}
     \centering
     \begin{subfigure}[b]{0.3\textwidth}
         \centering
         \begin{tikzpicture}[inner sep=0.7mm, scale = 0.25]
    \node[draw, circle, line width=1pt, fill=black] (y') at (-3, 1) [label=below: $y'$]{};
    \node[draw, circle, line width=1pt, fill=black] (y) at (0, 1) [label=below: $y$]{};
    \node[draw, circle, line width=1pt, fill=black] (x) at (0, 11) [label=above: $x$]{};
    \node[draw, circle, line width=1pt, fill=black] (x') at (-3, 11) [label=below: $x'$]{};
    \node[draw, circle, line width=1pt, fill=black] (z1) at (4, 6) [label=below: $z_1$]{};
    \node[draw, circle, line width=1pt, fill=black] (z1') at (7, 6) [label=below: $z_1'$]{};
    \node[draw, circle, line width=1pt, fill=black] (zj) at (12, 6) [label=below: $z_2$]{};
    \node[draw, circle, line width=1pt, fill=black] (zj') at (15, 6) [label=below: $z_2'$]{};
    \draw[black] (x) -- (x');
    \draw[black, middlearrow={<}] (x) -- (y);
    \draw[black] (y) -- (y');
    \draw[black, middlearrow={<}] (y) -- (z1);
    \draw[black, middlearrow={<}] (z1) -- (z1');
    \draw[red, middlearrow={>}] (x) -- (z1);
    \draw[red, middlearrow={<}] (y) -- (zj);
    \draw[black, middlearrow={>}] (zj) -- (zj');
    \draw[black, middlearrow={>}] (x) -- (zj);
    \end{tikzpicture}
         \caption{$z_1$ is of type $(+,-,+)$ and $z_2$ is of type $(+,-,-)$.}
         \label{fig:case1}
     \end{subfigure}
     \hfill
     \begin{subfigure}[b]{0.3\textwidth}
         \centering
         \begin{tikzpicture}[inner sep=0.7mm, scale = 0.25]
    \node[draw, circle, line width=1pt, fill=black] (y') at (-3, 1) [label=below: $y'$]{};
    \node[draw, circle, line width=1pt, fill=black] (y) at (0, 1) [label=below: $y$]{};
    \node[draw, circle, line width=1pt, fill=black] (x) at (0, 11) [label=above: $x$]{};
    \node[draw, circle, line width=1pt, fill=black] (x') at (-3, 11) [label=below: $x'$]{};
    \node[draw, circle, line width=1pt, fill=black] (z1) at (4, 6) [label=below: $z_1$]{};
    \node[draw, circle, line width=1pt, fill=black] (z1') at (7, 6) [label=below: $z_1'$]{};
    \node[draw, circle, line width=1pt, fill=black] (zj) at (12, 6) [label=below: $z_2$]{};
    \node[draw, circle, line width=1pt, fill=black] (zj') at (15, 6) [label=below: $z_2'$]{};
    \draw[black, middlearrow={<}] (x) -- (x');
    \draw[black, middlearrow={<}] (x) -- (y);
    \draw[black, middlearrow={>}] (y) -- (y');
    \draw[red, middlearrow={>}] (y) -- (z1);
    \draw[black, middlearrow={>}] (z1) -- (z1');
    \draw[black, middlearrow={>}] (x) -- (z1);
    \draw[black, middlearrow={<}] (y) -- (zj);
    \draw[black, middlearrow={<}] (zj) -- (zj');
    \draw[red, middlearrow={<}] (x) -- (zj);
    \end{tikzpicture}
         \caption{$z_1$ is of type $(+,+,-)$ and $z_2$ is of type $(-,-,+)$.}
         \label{fig:case2}
     \end{subfigure}
     \hfill
     \begin{subfigure}[b]{0.3\textwidth}
         \centering
         \begin{tikzpicture}[inner sep=0.7mm, scale = 0.25]
    \node[draw, circle, line width=1pt, fill=black] (y') at (-3, 1) [label=below: $y'$]{};
    \node[draw, circle, line width=1pt, fill=black] (y) at (0, 1) [label=below: $y$]{};
    \node[draw, circle, line width=1pt, fill=black] (x) at (0, 11) [label=above: $x$]{};
    \node[draw, circle, line width=1pt, fill=black] (x') at (-3, 11) [label=below: $x'$]{};
    \node[draw, circle, line width=1pt, fill=black] (z1) at (12, 6) [label=below: $z_i$]{};
    \node[draw, circle, line width=1pt, fill=black] (z1') at (15, 6) [label=below: $z_i'$]{};
    \draw[black, middlearrow={<}] (x) -- (x');
    \draw[red, middlearrow={<}] (x) -- (y);
    \draw[black, middlearrow={<}] (y) -- (y');
    \draw[black, middlearrow={>}] (y) -- (z1);
    \draw[black, middlearrow={>}] (z1) -- (z1');
    \draw[black, middlearrow={>}] (x) -- (z1);
    \end{tikzpicture}
         \caption{All $z_i$'s are of type $(+,+,-)$.}
         \label{fig:case3}
     \end{subfigure}
     \begin{subfigure}[b]{0.3\textwidth}
         \centering
         \begin{tikzpicture}[inner sep=0.7mm, scale = 0.25]
    \node[draw, circle, line width=1pt, fill=black] (y') at (-3, 1) [label=below: $y'$]{};
    \node[draw, circle, line width=1pt, fill=black] (y) at (0, 1) [label=below: $y$]{};
    \node[draw, circle, line width=1pt, fill=black] (x) at (0, 11) [label=above: $x$]{};
    \node[draw, circle, line width=1pt, fill=black] (x') at (-3, 11) [label=below: $x'$]{};
    \node[draw, circle, line width=1pt, fill=black] (z1) at (4, 6) [label=below: $z_i$]{};
    \node[draw, circle, line width=1pt, fill=black] (z1') at (7, 6) [label=below: $z_i'$]{};
    \node[draw, circle, line width=1pt, fill=black] (zj) at (12, 6) [label=below: $z_k$]{};
    \node[draw, circle, line width=1pt, fill=black] (zj') at (15, 6) [label=below: $z_k'$]{};
    \draw[black, middlearrow={<}] (x) -- (x');
    \draw[red, middlearrow={<}] (x) -- (y);
    \draw[black, middlearrow={>}] (y) -- (y');
    \draw[red, middlearrow={>}] (y) -- (z1);
    \draw[black, middlearrow={>}] (z1) -- (z1');
    \draw[black, middlearrow={>}] (x) -- (z1);
    \draw[black, middlearrow={<}] (y) -- (zj);
    \draw[black, middlearrow={<}] (zj) -- (zj');
    \draw[black, middlearrow={>}] (x) -- (zj);
    \end{tikzpicture}
         \caption{$z_k$ is of type $(+,-,+)$ and all other $z_i$'s are of type $(+,+,-)$.}
         \label{fig:case4}
     \end{subfigure}
     \hspace{5mm}
     \begin{subfigure}[b]{0.3\textwidth}
         \centering
         \begin{tikzpicture}[inner sep=0.7mm, scale = 0.25]
    \node[draw, circle, line width=1pt, fill=black] (y') at (-3, 1) [label=below: $y'$]{};
    \node[draw, circle, line width=1pt, fill=black] (y) at (0, 1) [label=below: $y$]{};
    \node[draw, circle, line width=1pt, fill=black] (x) at (0, 11) [label=above: $x$]{};
    \node[draw, circle, line width=1pt, fill=black] (x') at (-3, 11) [label=below: $x'$]{};
    \node[draw, circle, line width=1pt, fill=black] (z1) at (4, 6) [label=below: $z_i$]{};
    \node[draw, circle, line width=1pt, fill=black] (z1') at (7, 6) [label=below: $z_i'$]{};
    \node[draw, circle, line width=1pt, fill=black] (zj) at (12, 6) [label=below: $z_k$]{};
    \node[draw, circle, line width=1pt, fill=black] (zj') at (15, 6) [label=below: $z_k'$]{};
    \draw[black, middlearrow={<}] (x) -- (x');
    \draw[red, middlearrow={<}] (x) -- (y);
    \draw[black, middlearrow={>}] (y) -- (y');
    \draw[black, middlearrow={<}] (y) -- (z1);
    \draw[black, middlearrow={>}] (z1) -- (z1');
    \draw[black, middlearrow={>}] (x) -- (z1);
    \draw[black, middlearrow={>}] (y) -- (zj);
    \draw[black, middlearrow={>}] (zj) -- (zj');
    \draw[black, middlearrow={>}] (x) -- (zj);
    \end{tikzpicture}
         \caption{$z_k$ is of type $(+,-,-)$ and all other $z_i$'s are of type $(+,+,-)$.}
         \label{fig:case5}
     \end{subfigure}
        \caption{Some of the different possibilities for the orientation of arcs in $\overrightarrow{G_j}$.}
        \label{fig:meg<mag-cases}
\end{figure}

\begin{theorem}\label{thm:meg<mag-}
    Let $j \geq 1$, be an integer and let $G_j$ be as per Construction~\ref{cons:mag<mag-}. Then $mag^-(G_j) = j + 3$.
\end{theorem}
\begin{proof}
    Proposition~\ref{prop:SnS} tells us that the vertices $x', y'$ and $z_i', i \in \{1, \ldots, j\}$ will be in every MAG-set irrespective of the orientation of $G_j$. Hence, $mag^-(G_j) \geq j + 2$. Let us assume that there is an orientation $\overrightarrow{G_j}$ of $G_j$ such that $mag(\overrightarrow{G_j}) = j + 2$ and the MAG-set is the set of degree one vertices. We try to determine the properties of this orientation, if it exists.

    Without loss of generality, let $\overrightarrow{yx} \in A(\overrightarrow{G_j})$. Note that the vertex $z_i$ has exactly three neighbors, $x, y$, and $z'_i$. Suppose, $z$ is in $N^\alpha(x), N^\beta(y)$ and $N^\gamma(z_i')$ for some $\alpha, \beta, \gamma \in \{+,-\}$. In such a scenario, we say that $z_i$ is of type $(\alpha, \beta, \gamma)$. 
    
    If $z_i$ is of the type $(+,+,+)$, $(-,+,+)$, $(-,+,-)$, or $(-,-,-)$, then $z_i$ satisfies one of the conditions of Theorem~\ref{theorem:mageqn}, and thus is part of any MAG-set of $\overrightarrow{G_j}$. Hence, $z_i$ can only be of one of the types $(+,+,-)$, $(+,-,+)$, $(+,-,-)$, or $(-,-,+)$.

    \noindent \textit{Claim: } $\overrightarrow{G_j}$ can have at most one $z_i$ whose type is in $\{(+,-,+),(+,-,-)\}$.
    
\textit{Proof of claim:} Without loss of generality, suppose that $z_1$ is of type $(+,-,+)$ and $z_2$ is of type $(+,-,-)$ in $\overrightarrow{G_j}$ as shown in Figure~\ref{fig:case1}. It is clear that any shortest path through the arc $\overrightarrow{xz_1}$, whose endpoints monitor this arc, is of the form $\overrightarrow{P} = a \ldots x z_1 y \ldots b$, where $a$ can take one of the values $x'$ or $z_i', i \neq 1, 2$ and $b$ can take one of the values $y'$ or $z_i', i \neq 1, 2$. However, the existence of another path $\overrightarrow{P'} = a \ldots x z_2 y \ldots b$ which is of the same length as $\overrightarrow{P}$ ensures that one of the arcs $\overrightarrow{xz_1}$ or $\overrightarrow{z_2y}$ is never monitored in $\overrightarrow{G_j}$, which is a contradiction. Similarly, we can also show that $\overrightarrow{G_j}$ cannot have two $z_i$'s of type $(+,-,+)$ or $(+,-,-)$, respectively. \hfill$\diamond$

    \noindent \textit{Claim: } $\overrightarrow{G_j}$ cannot have $z_i$'s of types $(+,+,-)$ and $(-,-,+)$ simultaneously.\\
\noindent \textit{Proof of claim: } Without loss of generality, suppose that $z_1$ is of type $(+,+,-)$ and $z_2$ is of type $(-,-,+)$ in $\overrightarrow{G_j}$ as shown in Figure~\ref{fig:case2}. Observe that any shortest path between vertices that monitor $\overrightarrow{xz_1}$ is of the form $P = a \ldots x z_1 z_1'$. The vertex $a$ cannot be $y$ since there is a shorter path from $y$ to $z_1'$. The vertex $a$ cannot be $z_2'$ since there is another shortest path $P' = z_2' z_2 y z_1 z_1'$ not passing through $\overrightarrow{xz_1}$. If $a$ is an arbitrary $z_i'$, then the corresponding $z_i$ will be part of every MAG-set by Theorem~\ref{theorem:mageqn}. Hence, $a$ must be $x'$ and hence, $\overrightarrow{x'x} \in A(\overrightarrow{G_j})$. Similarly, to monitor the arc $\overrightarrow{z_2y}$, we must have $\overrightarrow{yy'} \in A(\overrightarrow{G_j})$. But now to monitor the arc $\overrightarrow{yz_1}$, $\overrightarrow{G_j}$ must contain a $z_i$ of type $(+,-,+)$ and to monitor the arc $\overrightarrow{z_2x}$, $\overrightarrow{G_j}$ must contain a $z_i$ of type $(+,-,-)$, which is a contradiction. \hfill$\diamond$

So the graph $\overrightarrow{G_j}$ must have one of the following properties.
\begin{enumerate}[(i)]
    \item All $z_i$'s in $\overrightarrow{G_j}$ are of type $(+,+,-)$.
    \item In $\overrightarrow{G_j}$, for some particular $k$, $z_k$ is of type $(+,-,+)$ and all other $z_i$'s are of type $(+,+,-)$.
    \item In $\overrightarrow{G_j}$, for some particular $k$, $z_k$ is of type $(+,-,-)$ and all other $z_i$'s are of type $(+,+,-)$.
    \item All $z_i$'s in $\overrightarrow{G_j}$ are of type $(-,-,+)$.
    \item In $\overrightarrow{G_j}$, for some particular $k$, $z_k$ is of type $(+,-,+)$ and all other $z_i$'s are of type $(-,-,+)$.
    \item In $\overrightarrow{G_j}$, for some particular $k$, $z_k$ is of type $(+,-,-)$ and all other $z_i$'s are of type $(-,-,+)$.
\end{enumerate}

In (i), the orientation of the edges $x'x$ and $y'y$ is forced as $\overrightarrow{x'x}$ and $\overrightarrow{y'y}$ to ensure that the arcs $\overrightarrow{xz_i}$ and $\overrightarrow{yz_i}$ are monitored. But then, $\overrightarrow{yx}$ is not monitored by any of the pairs of vertices in the MAG-set. This is illustrated in Figure~\ref{fig:case3}. In (ii), the arc $\overrightarrow{xz_k}$ is monitored only if $\overrightarrow{x'x}, \overrightarrow{yy'} \in A(\overrightarrow{G_j})$ as shown in Figure~\ref{fig:case4}. If there is no $z_i$ of type $(+,+,-)$ in $\overrightarrow{G_j}$, then the arc $\overrightarrow{yx}$ is not monitored and if there is at least one $z_i$ of type $(+,+,-)$ in $\overrightarrow{G_j}$, then the arc $\overrightarrow{yz_i}, i \neq k$ is never monitored. Hence, (ii) is not possible. Observe that in (iii), $z_k'$ does not help in monitoring any of the arcs in the neighborhood of the $z_i$'s, $i \neq k$, if they exist. Monitoring the arc $\overrightarrow{z_ky}$ happens only if $\overrightarrow{x'x}, \overrightarrow{yy'} \in A(\overrightarrow{G_j})$. But then the arc $\overrightarrow{yx}$ is never monitored. This is illustrated in Figure~\ref{fig:case5}. The cases (iv), (v) and (vi) are symmetric to (i), (ii) and (iii). Similar arguments as above deny their possibilities too.

Thus, we have proved that the graph $\overrightarrow{G_j}$ with $mag(\overrightarrow{G_j}) = j + 2$ does not exist. Moreover, if $\overrightarrow{G_j}$ consists solely of $z_i$'s of type $(+,+,-)$, then adding the vertex $x$ to the MAG-set ensures that all the arcs are monitored. Hence, $mag^-(G_j) = j + 3$.
\end{proof}

The following corollary is a direct interpretation of Construction~\ref{cons:mag<mag-} and Theorem~\ref{thm:meg<mag-}.

\begin{corollary}
\label{cor:mag<mag-}
    For any positive integer $j$, where $j \geq 3$, we can find a connected graph $G_j$ such that $meg(G_j) = j + 2$ and $meg(G_j) < mag^-(G_j)$.
\end{corollary}

Hence, there does not exist any trivial relation between the monitoring edge-geodetic number of a graph and the monitoring arc-geodetic number of any of its orientations.

\section{MAG-sets of tournaments}
\label{sec4}
\label{sec:tournaments}
In this section, we will study \emph{tournaments} (oriented graphs where the underlying graph is complete), and give a complete characterization of our notion on them. In particular, the result from the next theorem indicates that there are only two possibilities, any MAG-set consists of either all the vertices of the tournament or at most one vertex may be omitted. 
\begin{theorem}
\label{thm:tornchar}
    Let $\overrightarrow{G}$ be a tournament of order $n$. Then $mag(\overrightarrow{G}) \geq n-1$.
\end{theorem}
\begin{proof}
    Let $\overrightarrow{G}$ be a tournament of order $n$. Let us assume, on the contrary, that $mag(\overrightarrow{G})$ $\leq n - 2$. Then, for any MAG-set $M$ of $\overrightarrow{G}$ of size $mag(\overrightarrow{G})$, there exist some vertices $u,v \in V(\overrightarrow{G})$ which are not a part of $M$. We assume without loss of generality that $\overrightarrow{uv}$ is an arc of $\overrightarrow{G}$. Be definition, there exist some vertices $a, b$ in $M$ such that $a$ and $b$ monitor $\overrightarrow{uv}$. Consider a shortest path $\overrightarrow{P}$ from $a$ to $b$ containing $\overrightarrow{uv}$. Without loss of generality, let $\overrightarrow{P}$ be the path given by $\overrightarrow{P} = a \ldots x u v y \ldots b$ (possibly, $a=x$ or $b=y$). Since $\overrightarrow{P}$ is a shortest path from $a$ to $b$, we know that $\overrightarrow{vx}$ and $\overrightarrow{yu}$ are arcs of $\overrightarrow{G}$.
    
    Let $c, d$ be a pair of distinct vertices in $M$ which monitor the arc $\overrightarrow{yu}$. Consider a shortest path $\overrightarrow{P'}$ between $c$ and $d$ containing the arc $\overrightarrow{yu}$. Without loss of generality, we can write $\overrightarrow{P'}$ as $\overrightarrow{P'} = c \ldots c' y u d' \ldots d$ (possibly, $c=c'$, $c'=y$, $u=d'$, or $d'=d$). Observe that, although $c$ and $y$ could be the same vertex, it cannot be that $d$ and $u$ are the same vertex because $u$ is not a part of $M$. Since $\overrightarrow{P'}$ is a shortest path, it must be that $\overrightarrow{d'y}$ is an arc of $\overrightarrow{G}$. In this case, the path $a \ldots x u d' y \ldots, b$ where the vertices in the subpaths $a \ldots x$ and $y \ldots b$ are identical to those of $\overrightarrow{P}$, is a path from $a$ to $b$ of the same length as $\overrightarrow{P}$, hence is also a shortest path. The fact that this shortest path between $a$ and $b$ does not contain the arc $\overrightarrow{uv}$ contradicts the fact that $\overrightarrow{uv}$ is monitored by $a$ and $b$. Hence, $mag(\overrightarrow{G})$ $\geq n - 1$. 
\end{proof}
\begin{figure}
    \centering
    \begin{tikzpicture}[inner sep=0.7mm, scale = 1]
    \node[draw, circle, line width=1pt, fill=black] (x) at (0, 0) [label=below: $x$]{};
    \node[draw, circle, line width=1pt, fill=black] (u) at (-1, 1) [label=below left: $u$]{};
    \node[draw, circle, line width=1pt, fill=black] (v) at (1, 1) [label=right: $v$]{};
    \node[draw, circle, line width=1pt, fill=black] (y) at (0, 2) [label=above right: $y$]{};
    \node[draw, circle, line width=1pt, fill=black] (d') at (-2, 1) [label=above: $d'$]{};
    \draw[-stealth, red] (x) -- (u);
    \draw[-stealth, black] (u) -- (v);
    \draw[-stealth, black] (v) -- (x);
    \draw[-stealth, black] (v) -- (y);
    \draw[-stealth, black] (y) -- (u);
    \draw[-stealth, red] (d') -- (y);
    \draw[-stealth, red] (u) -- (d');
    \end{tikzpicture}
    \caption{When $u$ and $v$ are not part of an MAG-set, $\protect\overrightarrow{uv}$ is not monitored.}
    \label{fig:my_label}
\end{figure}

To show that the bound in Theorem~\ref{thm:tornchar} is tight, the next two results show that there exist orientations of any tournament which achieve both the possible sizes of MAG-sets.

\begin{theorem}
\label{thm:mag+}
    If $K_n$ is a complete graph of order $n \geq 3$, then $mag^+(K_n) = n$.
\end{theorem}
\begin{proof}
    Let $\overrightarrow{K_n}$ be a transitive tournament with vertex set $V(\overrightarrow{K_n}) = \{v_1, \ldots, v_n\}$ and arc set $A(\overrightarrow{K_n}) = \{\overrightarrow{v_i v_j} \colon\ i,j \in [1, n] \text{ and } i < j\}$. 
    Clearly, $v_1$ is a source and $v_n$ is a sink in $\overrightarrow{K_n}$. Further, observe that for any vertex $v_i \in V(\overrightarrow{K_n}), i \in [2, n - 1]$, $v_{i - 1}$ is adjacent to all $v_j \in N^+(v_i)$, since $i - 1 < i < j$. Therefore, by Theorem~\ref{theorem:mageqn} (ii), all vertices of $\overrightarrow{K_n}$ are a part of every MAG-set and hence, $mag^+(K_n) = n$.
\end{proof}

\begin{theorem}
    If $K_n$ is a complete graph of order $n \geq 3$, then $mag^-(K_n) = n - 1$.
\end{theorem}
\begin{proof}
    Let us consider a transitive tournament $\overrightarrow{K_n}$ as in the proof of Theorem~\ref{thm:mag+}. Let $\overrightarrow{K_n'}$ be the graph obtained by flipping the arcs $\overrightarrow{v_{i} v_n}$, for $i \in \{1, \ldots, n - 2\}$ in $\overrightarrow{K_n}$, that is $\overrightarrow{K_n'}$ is the graph obtained by replacing the arcs $\overrightarrow{v_{i} v_n}$ of $\overrightarrow{K_n}$ with the arcs $\overrightarrow{v_n v_{i}}$. We claim that this tournament has an MAG-set of size $n - 1$. In particular, we claim that $\{v_1, \ldots, v_{n - 1}\}$ is an MAG-set of $\overrightarrow{K_n'}$. To see this, first observe that the arcs $\overrightarrow{v_i v_j}, i,j \in \{1, \ldots, n - 1\},$ with $i < j$ are all monitored by the vertices incident to them. Similarly, for $i \in \{1, \ldots, n - 2\}$, the vertices $v_i$ and $v_{n - 1}$ monitor the arcs $\overrightarrow{v_n v_i}$ and $\overrightarrow{v_{n - 1} v_n}$ as they lie on the unique shortest path $v_{n - 1} v_n v_i$ from $v_{n - 1}$ to $v_i$. Thus, all arcs of $\overrightarrow{K_n'}$ are monitored, and hence $mag(\overrightarrow{K_n'}) \leq n - 1$. Together with Theorem~\ref{thm:tornchar}, we have that $mag(\overrightarrow{K_n'}) = n - 1$. Hence, $mag^-(K_n) = n - 1$ and this concludes the proof.
\end{proof}

Chartrand and Zhang~\cite{chartrand2000geodetic} proved that the geodetic number of a tournament of order $n$ can take any value in the set $\{2, \ldots, n\}$. The above results for the monitoring arc-geodetic number of tournaments show a surprising contrast from the result for the geodetic number of tournaments. 

Note that checking whether a subset of vertices is a monitoring arc set can be done in polynomial time. Hence, we can enumerate in polynomial time all the MAG-sets of a given tournament by testing the $n$ subsets of vertices of size $n - 1$, where $n$ is the order of the tournament. But finding the MAG-sets of general graphs is not so easy as evidenced by the next section.

\section{Computational hardness of determining the monitoring arc-geodetic number}
\label{sec5}

In this section, we study the algorithmic complexity of two problems -- the \textsc{MAG$^+$-set} problem and the \textsc{$k$-MAG-set} problem.

\subsection{\textsc{MAG$^+$-set} problem}
For an undirected graph $G$, the orientation $\overrightarrow{G}$ of $G$ with $mag(\overrightarrow{G}) = mag^-(G)$ is obtained intuitively, by reducing the number of sources and sinks when orienting the graph. Of course, reducing the number of sources and sinks does not necessarily reduce the monitoring arc-geodetic number as shown by Theorem~\ref{theorem:mageqn}. An interesting observation is that $mag^-(G) = n$, only when it is not possible to orient $G$ such that it has a vertex that is neither a source nor a sink. Such a graph $G$ is trivially $K_2$. 

The study of the possible values that the lower monitoring arc-geodetic number can take, makes for an interesting problem, but we shall focus on the family of graphs for which $mag^+(G) = n$.

Observe that if $mag^+(G) \neq n$, then $G$ must contain an odd cycle. If not, then $G$ is bipartite, and we can orient the edges in such a way that the vertices of one part are sources and that of the other are sinks. However, identifying any other characteristics of graphs $\overrightarrow{G}$ with $mag^+(\overrightarrow{G}) \neq n$ is not so easy. In fact, let us consider the following decision problem:

\bigskip
\begin{mdframed}
\noindent
\textsc{MAG$^+$-set} problem

\noindent
\textbf{Instance:} An undirected graph $G$.

\noindent
\textbf{Question:} Is $mag^+(G) = |V(G)|$?
\end{mdframed}

In what follows, we prove that the \textsc{MAG$^+$-set} problem is NP-hard. To do this, we give a reduction from the monotone not-all-equal 3-satisfiability (\textsc{NAE $3$-SAT}) problem.

An instance $\Phi$ of the \textsc{NAE $3$-SAT} problem is given by a $3$-SAT formula with no negative literals and no literals appearing twice in a same clause, and is satisfiable if every clause contains at least one literal valued at True and at least one literal valued at False. This problem can be proven to be NP-hard, for example, using a reduction from \textsc{$3$-SAT}~\cite{moore2011nature}. 

\begin{construction}
\label{cons:mag+hard}
    We denote by $\{x_1, \ldots, x_n\}$, the literals of an instance $\Phi$ of the monotone \textsc{NAE 3-SAT} problem, and $\{C_1, \ldots, C_m\}$ its clauses. We now construct the graph $G(\Phi)$ with each literal associated with a vertex and each clause associated with a triangle. To be precise, $G(\Phi)$ is the graph defined with $V(G(\Phi)) = \{x_1, \ldots, x_n\} \cup \{c_{1, 1}, \ldots, c_{m, 1}\} \cup \{c_{1, 2}, \ldots, c_{m, 2}\} \cup \{c_{1, 3}, \ldots, c_{m, 3}\}$ and with edges $c_{i, 1}c_{i, 2}, c_{i, 1}c_{i, 3}, c_{i, 2}c_{i, 3}$ for every $i, 1\leq i\leq m$, and the edge $x_ic_{j, k}$ if $x_i$ is the $k$-th literal of the clause $C_j$ for every $i, j, k$ where $1\leq i\leq n, 1\leq j\leq m$ and $1\leq k \leq 3$. This construction is illustrated in Figure~\ref{fig:naeset}. It can easily be done in polynomial time. 
\end{construction}

\begin{figure}
    \centering
    \begin{tikzpicture}[inner sep=0.7mm, scale = 1]
        \node[draw, circle, line width=1pt, fill=black] (c11) at (0, 0) [label=below right: $c_{1,1}$]{};
        \node[draw, circle, line width=1pt, fill=black] (c12) at (2, 0) [label=below right: $c_{1,2}$]{};
        \node[draw, circle, line width=1pt, fill=black] (c13) at (1, 1.5) [label=above: $c_{1,3}$]{};
        \node[draw, circle, line width=1pt, fill=black] (c21) at (0, -3) [label=below right: $c_{2,1}$]{};
        \node[draw, circle, line width=1pt, fill=black] (c22) at (2, -3) [label=below right: $c_{2,2}$]{};
        \node[draw, circle, line width=1pt, fill=black] (c23) at (1, -1.5) [label=above: $c_{2,3}$]{};
        \node[draw, circle, line width=1pt, fill=black] (cm1) at (0, -8) [label=below right: $c_{m,1}$]{};
        \node[draw, circle, line width=1pt, fill=black] (cm2) at (2, -8) [label=below right: $c_{m,2}$]{};
        \node[draw, circle, line width=1pt, fill=black] (cm3) at (1, -6.5) [label=above: $c_{m,3}$]{};
        \node[draw, circle, line width=1pt, fill=black] (x1) at (-5, 1.5) [label=below: $x_1$]{};
        \node[draw, circle, line width=1pt, fill=black] (x2) at (-5, -0.5) [label=below: $x_2$]{};
        \node[draw, circle, line width=1pt, fill=black] (xi) at (-5, -4) [label=below: $x_i$]{};
        \node[draw, circle, line width=1pt, fill=black] (xn) at (-5, -8) [label=below: $x_n$]{};\
        \node[] (xn1) at (-5, -7){\Huge.};
        \node[] (xn2) at (-5, -6){\Huge.};
        \node[] (xn3) at (-5, -3){\Huge.};
        \node[] (xn4) at (-5, -2){\Huge.};
        \node[] (xn5) at (1, -4){\Huge.};
        \node[] (xn6) at (1, -5){\Huge.};
        \node[] (xn7) at (1, -5.75){\Huge.};
        \node[] (xn8) at (-5, -5){\Huge.};
        \node[] (xn9) at (-5, -1.25){\Huge.};
        \draw[] (c11) -- (c12);
        \draw[] (c12) -- (c13);
        \draw[] (c13) -- (c11);
        \draw[] (c21) -- (c22);
        \draw[] (c22) -- (c23);
        \draw[] (c23) -- (c21);
        \draw[] (cm1) -- (cm2);
        \draw[] (cm2) -- (cm3);
        \draw[] (cm3) -- (cm1);
        \draw[] (xi) -- (cm1);
        \draw[] (xi) -- (c23);
        \draw[] (xi) -- (c11);
    \end{tikzpicture}
    \caption{Construction of the graph $G(\Phi)$.}
    \label{fig:naeset}
\end{figure}
\begin{lemma}
\label{lem:mag+hard1}
    If $mag^+(G(\Phi)) = |V(G(\Phi))| = n + 3m$, then $\Phi$ is satisfiable.
\end{lemma}
\begin{proof}
    Let $mag^+(G(\Phi)) = |V(G(\Phi))| = n + 3m$. Hence, there exists an orientation $\overrightarrow{G}$ of $G(\Phi)$ such that $mag(\overrightarrow{G}) = V(\overrightarrow{G}) = n + 3m$.

    \noindent \textit{Claim:} $x_i$ is either a source or a sink in $\overrightarrow{G}$, for all $1 \leq i \leq n$.\\ 
    \noindent \textit{Proof of the claim:} Since no literal can appear more than once in a single clause, $x_i$ cannot be a part of any cycle of length 3 or 4 in $G(\Phi)$. Hence, by Theorem \ref{theorem:mageqn}, it must be that $x_i$ is either a source or a sink in $\overrightarrow{G}$. \hfill\mbox{$\diamond$}

    Let us denote by $x_{j, 1}, x_{j, 2}$ and $x_{j, 3}$, the associated literal vertices of the clause $C_j$, for $1 \leq j \leq m$.

    \noindent \textit{Claim:} For a fixed $j$, $1 \leq j \leq m$, not all of $x_{j, 1}, x_{j, 2}$ and $x_{j, 3}$ are sources, nor are all of them sinks.\\
    \noindent \textit{Proof of the claim:} Without loss of generality, suppose that all the three of $x_{j, 1}, x_{j, 2}$ and $x_{j, 3}$ are sources. Up to renaming, there are only two possible orientations of the cycle $c_{j, 1}c_{j, 2}c_{j, 3}$. If the arcs of the cycle are $\overrightarrow{c_{j, 1}c_{j, 2}}, \overrightarrow{c_{j, 2}c_{j, 3}}$ and $\overrightarrow{c_{j, 3}c_{j, 1}}$, then the choice of any two vertices amongst $c_{j, 1}, c_{j, 2}$ and $c_{j, 3}$ is enough to monitor all the arcs of the cycle, as there is no shortest path through any of $x_{j, 1}, x_{j, 2}$ or $x_{j, 3}$. This contradicts the fact that $mag(\overrightarrow{G}) = V(\overrightarrow{G})$. Hence, without loss of generality, let us assume that $\overrightarrow{c_{j, 1}c_{j, 2}},\overrightarrow{c_{j, 2}c_{j, 3}}$ and $\overrightarrow{c_{j, 1}c_{j, 3}}$ are arcs of $\overrightarrow{G}$. In this case, it is again easy to check that $c_{j, 2}$ and $c_{j, 3}$ together with $x_{j, 1}, x_{j, 2}$ and $x_{j, 3}$ are enough to monitor all the arcs of the cycle, again contradicting the value of $mag(\overrightarrow{G})$. This proves the claim. \hfill\mbox{$\diamond$}
    
\medskip
Given an orientation $\overrightarrow{G}$ of $G(\Phi)$ with $mag^+(G(\Phi)) = |V(G(\Phi))|$, we can now provide a valuation of the literals of $\Phi$. Assign $x_i$ the value True if $x_i$ is a sink, and the value False, otherwise, for $1 \leq i \leq n$. We have proven that for $i \in \{1, 2, 3\}$, not all of $x_{j,i}$ are sources nor are all of them sinks. This implies that in a clause, not all literals have the same valuation. Hence, the valuation of the literals of $\Phi$ given in this way satisfies $\Phi$.
\end{proof}

\begin{lemma}
\label{lem:mag+hard2}
    If $\Phi$ is satisfiable, then $mag^+(G(\Phi)) = |V(G(\Phi))| = n + 3m$.
\end{lemma}
\begin{proof}
Let $C$ be a satisfiable truth assignment of $\Phi$. We construct the graph $G(\Phi)$ corresponding to $\Phi$ as per Construction~\ref{cons:mag+hard}. We obtain an orientation $\overrightarrow{G}$ of $G(\Phi)$ as follows. For all $i, 1 \leq i \leq n$, $x_i$ is oriented to be a sink if $X_i \in C$ is True, and it is oriented to be a source otherwise. Since the $x_i$'s are non-adjacent, this is always possible. 

Let us denote by $x_{j, 1}, x_{j, 2}$ and $x_{j, 3}$, the associated literal vertices of the clause $c_{j, 1}, c_{j, 2}, c_{j, 3}$, for any $j, 1 \leq j \leq m$. Since $C$ satisfies $\Phi$, not all of $x_{j, 1}, x_{j, 2}$ and $x_{j, 3}$ are sources, nor are all of them sinks. Without loss of generality, let us assume that $x_{j, 1}$ is a sink, and both $x_{j, 2}$ and $x_{j, 3}$ are sources. We orient the arcs of each cycle $c_{j, 1}, c_{j, 2}, c_{j, 3}$ as $\overrightarrow{c_{j, 1}c_{j, 2}}$, $\overrightarrow{c_{j, 1}c_{j, 3}}$ and $\overrightarrow{c_{j, 2}c_{j, 3}}$. We now claim that $mag(\overrightarrow{G}) = n + 3m$. By Proposition~\ref{prop:SnS}, the vertices $x_i$, where $1 \leq i \leq n$, $c_{j, 1}$ and $c_{j, 3}$, where $1 \leq j \leq m$, are always a part of every MAG-set. Furthermore, observe that $N^+(c_{j, 2}) = \{c_{j, 3}\}$ and that $c_{j, 1} \in N^-(c_{j, 2})$ is a vertex such that $\overrightarrow{c_{j, 1}c_{j, 3}} \in A(\overrightarrow{G})$. Hence, by Theorem~\ref{theorem:mageqn} (ii), $c_{j, 2}$ is also a part of every MAG-set and $mag^+(G(\Phi)) = mag(\overrightarrow{G}) = n + 3m$.
\end{proof}

\begin{theorem}\label{thm:hardmag+}
    \textsc{MAG$^+$-set} is NP-hard.
\end{theorem}
\begin{proof}
    Lemmas~\ref{lem:mag+hard1} and \ref{lem:mag+hard2} prove that an instance of monotone \textsc{NAE 3-SAT} of size $n + m$ can be reduced to an instance of \textsc{MAG$^+$-set} problem of size $n + 3m$. Hence, the problem is NP-hard.
\end{proof}

This result actually underlines the impact of the criteria expressed in Theorem \ref{theorem:mageqn}. Indeed, the easiest way to increase the monitoring arc-geodetic number would be to increase the number of sources and sinks. However, rendering a graph MAG-extremal only with sources and sinks implies that the graph is bipartite, and this can be checked in polynomial time. The statement of Theorem \ref{thm:hardmag+} instead proves that identifying the worst choice of arcs for odd cycles in non-oriented graphs is indeed a hard problem.

\subsection{\textsc{$k$-MAG-set} problem}
Haslegrave~\cite{HASLEGRAVE202379} proved the following theorem for the monitoring edge-geodetic set of graphs:

\begin{theorem}
    The decision problem of determining for a graph $G$ and natural number $k$ whether $meg(G) \leq k$ is NP-complete.
\end{theorem}

In what follows, we prove that the oriented analogue of this problem is also NP-complete. Formally, we define MAG-set to be the following decision problem:

\bigskip
\begin{mdframed}[nobreak=true]
\noindent
\textsc{$k$-MAG-set} problem

\noindent
\textbf{Instance:} An oriented graph $\overrightarrow{G}$ and an integer $2 \leq k \leq |V(\overrightarrow{G})|$.

\noindent
\textbf{Question:} Does there exist an MAG-set of $\overrightarrow{G}$ of size $k$?
\end{mdframed}

\begin{figure}
    \centering
    \begin{tikzpicture}[inner sep=0.7mm, scale = 1]
        \node[draw, circle, line width=1pt, fill=black] (x0) at (0, 0) [label=below: $e_1$]{};
        \node[draw, circle, line width=1pt, fill=black] (x1) at (1, 0) [label=below: $e_2$]{};
        \node[draw, circle, line width=1pt, fill=black] (x2) at (2, 0) [label=below: $e_3$]{};
        \node[draw, circle, line width=1pt, fill=black] (x3) at (3, 0) [label=below: $e_4$]{};
        \node[draw, circle, line width=1pt, fill=black] (x4) at (4, 0) [label=below: $e_5$]{};
        \node[draw, circle, line width=1pt, fill=black] (S0) at (0, 2) [label=above: $v_1$]{};
        \node[draw, circle, line width=1pt, fill=black] (S1) at (1, 2) [label=above: $v_2$]{};
        \node[draw, circle, line width=1pt, fill=black] (S2) at (2, 2) [label=above: $v_3$]{};
        \node[draw, circle, line width=1pt, fill=black] (S3) at (3, 2) [label=above: $v_4$]{};
        \node[draw, circle, line width=1pt, fill=black] (S4) at (4, 2) [label=above: $v_5$]{};
        \draw[] (S0) -- (x0);
        \draw[] (S1) -- (x0);
        \draw[] (S1) -- (x3);
        \draw[] (S1) -- (x4);
        \draw[] (S2) -- (x2);
        \draw[] (S2) -- (x4);
        \draw[] (S3) -- (x1);
        \draw[] (S3) -- (x3);
        \draw[] (S4) -- (x1);
        \draw[] (S4) -- (x2);
    \end{tikzpicture}
    \caption{A representation of an instance of \textsc{Vertex Cover}}
    \label{fig:exsetcover}
\end{figure}
\begin{figure}
    \centering
    \begin{tikzpicture}[inner sep=0.7mm, scale = 1]
        \node[draw, circle, line width=1pt, fill=black] (x0) at (0, 0) [label=below right: $e_1$]{};
        \node[draw, circle, line width=1pt, fill=black!30] (x1) at (1, 0) [label=below right: $e_2$]{};
        \node[draw, circle, line width=1pt, fill=black] (x2) at (2, 0) [label=below right: $e_3$]{};
        \node[draw, circle, line width=1pt, fill=black!30] (x3) at (3, 0) [label=below right: $e_4$]{};
        \node[draw, circle, line width=1pt, fill=black!30] (x4) at (4, 0) [label=below right: $e_5$]{};
        
        \node[draw, circle, line width=1pt, fill=black] (e0) at (-.25, -1) [label=below: $g_1$]{};
        \node[draw, circle, line width=1pt, fill=black!30] (e1) at (.75, -1) [label=below: $g_2$]{};
        \node[draw, circle, line width=1pt, fill=black] (e2) at (1.75, -1) [label=below: $g_3$]{};
        \node[draw, circle, line width=1pt, fill=black!30] (e3) at (2.75, -1) [label=below: $g_4$]{};
        \node[draw, circle, line width=1pt, fill=black!30] (e4) at (3.75, -1) [label=below: $g_5$]{};

        \node[draw, circle, line width=1pt, fill=black] (f0) at (.25, -1) [label=below: $f_1$]{};
        \node[draw, circle, line width=1pt, fill=black!30] (f1) at (1.25, -1) [label=below: $f_2$]{};
        \node[draw, circle, line width=1pt, fill=black] (f2) at (2.25, -1) [label=below: $f_3$]{};
        \node[draw, circle, line width=1pt, fill=black!30] (f3) at (3.25, -1) [label=below: $f_4$]{};
        \node[draw, circle, line width=1pt, fill=black!30] (f4) at (4.25, -1) [label=below: $f_5$]{};

        \node[draw, circle, line width=1pt, fill=black] (S0) at (0, 2) [label=above left: $v_1$]{};
        \node[draw, circle, line width=1pt, fill=black!30] (S1) at (1, 2) [label=above: $v_2$]{};
        \node[draw, circle, line width=1pt, fill=black!30] (S2) at (2, 2) [label=above: $v_3$]{};
        \node[draw, circle, line width=1pt, fill=black!30] (S3) at (3, 2) [label=above: $v_4$]{};
        \node[draw, circle, line width=1pt, fill=black!30] (S4) at (4, 2) [label=above: $v_5$]{};
        
        \node[draw, circle, line width=1pt, fill=black] (a0) at (-.5, 4) [label=above: $a_1$]{};
        \node[draw, circle, line width=1pt, fill=black] (b0) at (.5, 4) [label=above: $b_1$]{};
        \node[draw, circle, line width=1pt, fill=black] (c0) at (0, 3) [label=left: $c_1$]{};
        \draw[-stealth, black] (a0) -- (c0);
        \draw[-stealth, black] (c0) -- (b0);
        \draw[-stealth, black] (c0) -- (S0);
        \draw[-stealth, black] (a0)to[in=180, out =210] (e0);
        \draw[-stealth, black] (a0)to[in=150, out =210] (e2);
        \draw[-stealth, black] (S0) -- (x0);
        \draw[-stealth, black] (S0) -- (x2);
        \draw[-stealth, black] (S0) -- (x0);
        \draw[-stealth, black] (S1) -- (x0);
        \draw[-stealth, black] (S1) -- (x3);
        \draw[-stealth, black] (S1) -- (x4);
        \draw[-stealth, black] (S2) -- (x2);
        \draw[-stealth, black] (S2) -- (x4);
        \draw[-stealth, black] (S3) -- (x1);
        \draw[-stealth, black] (S3) -- (x3);
        \draw[-stealth, black] (S4) -- (x1);
        \draw[-stealth, black] (S4) -- (x2);
                
        \draw[-stealth, black] (x0) -- (e0);
        \draw[-stealth, black!30] (x1) -- (e1);
        \draw[-stealth, black] (x2) -- (e2);
        \draw[-stealth, black!30] (x3) -- (e3);
        \draw[-stealth, black!30] (x4) -- (e4);
        \draw[-stealth, black] (x0) -- (f0);
        \draw[-stealth, black!30] (x1) -- (f1);
        \draw[-stealth, black] (x2) -- (f2);
        \draw[-stealth, black!30] (x3) -- (f3);
        \draw[-stealth, black!30] (x4) -- (f4);
    \end{tikzpicture}
    \caption{A transformation of an instance of \textsc{Vertex Cover} into an instance of $k$-MAG-set.}
    \label{fig:transVcover}
\end{figure}

\begin{theorem}
\label{thm:hardness}
    The \textsc{$k$-MAG-set} problem is NP-complete when restricted to acyclic oriented graphs with maximum degree $4$.
\end{theorem}

To prove this result, we give a reduction from the \textsc{Vertex Cover} decision problem, which is well known to be NP-complete~\cite{karp2010reducibility}. In fact, the \textsc{Vertex Cover} problem is known to be NP-complete even when restricted to graphs with maximum degree $3$~\cite{10.1145/800119.803884}.
The vertex cover problem is defined as follows: 

\bigskip
\begin{mdframed}[nobreak=true]
\noindent
\textsc{Vertex Cover}

\noindent
\textbf{Instance}: A connected graph $G$ and an integer $k$.

\noindent
\textbf{Question}: Does there exist a subset $M\subset V(G)$ of size at most $k$ such that every edge of $G$ has at least one endpoint in $M$?    
\end{mdframed}

\medskip
Let $\Phi$ be an instance of \textsc{Vertex Cover}, given as a list of edges $e_1, \dots, e_m$ and the integer $k$. It can be encoded as a bipartite graph $G(\Phi)$ with parts $T$ and $U$, with the edges represented by vertices in $U$, the $n$ distinct vertices of $\Phi$ represented by vertices in $T$, and the edges of $G(\Phi)$ represent the incidence of a vertex and an edge. For example, the graph of Figure \ref{fig:exsetcover} is an encoding of the \textsc{Vertex Cover} problem $\{v_1v_2, v_4v_5, v_3v_5, v_2v_4, v_2v_3\}$.

\begin{construction}
\label{cons:hardness}
    We modify this representation of the \textsc{Vertex Cover} problem $\Phi$ to obtain an instance $\overrightarrow{H(\Phi)}$ of the \textsc{$k$-MAG-set} problem. First, we orient all the arcs from the vertices representing vertices $v_i$ to the vertices representing edges $e_j$. Then, for each vertex $e_i$ representing an edge, we add two new vertices $g_i$ and $f_i$ and the arcs of the form $\overrightarrow{e_ig_i}$ and $\overrightarrow{e_if_i}$. Subsequently, corresponding to each vertex $v_i$ representing a vertex, we add three new vertices $a_i, b_i$ and $c_i$ and the arcs $\overrightarrow{a_ic_i}, \overrightarrow{c_ib_i}$ and $\overrightarrow{c_iv_i}$. Finally, for every arc of the form $\overrightarrow{v_ie_j}$ in the graph $\overrightarrow{H(\Phi)}$, we add the arc $\overrightarrow{a_ig_j}$. We denote by $A$, the set of all the vertices $a_i, b_i, g_i$ and $f_i$ (note that it does not include the vertices of the type $c_i$). The size of $A$ is $|A| = 2|V(G(\Phi))|$.
\end{construction}     

    As an example of this construction, Figure~\ref{fig:transVcover} is a partial transformation of the \textsc{Vertex Cover} problem shown in Figure~\ref{fig:exsetcover} to the \textsc{$k$-MAG-set} problem.
    
\begin{lemma}
\label{lem:hard1}
    Let $\Phi$ be an instance of the \textsc{Vertex Cover} problem with $n$ vertices and $m$ edges. If $\Phi$ has a solution set of size $k$, then $\overrightarrow{H(\Phi)}$ has an MAG-set of size $k + 2n + 2m$.
\end{lemma}
\begin{proof}
Let $\overrightarrow{H(\Phi)}$ be constructed from $G(\Phi)$ as described in Construction~\ref{cons:hardness}. Here, $|A| = 2n + 2m$ and all the vertices of $A$ are either sinks or sources. Hence, by Proposition~\ref{prop:SnS}, the vertices of $A$ will always be part of every MAG-set. This also implies that the arcs of the form $\overrightarrow{a_ic_i}, \overrightarrow{c_ib_i}$ are always monitored by $a_i$ and $b_i$, and so are the arcs of the form $\overrightarrow{a_ig_j}$, since they are incident to two elements of the MAG-set. The arcs of the form $\overrightarrow{c_iv_i}, \overrightarrow{v_ie_j}, \overrightarrow{e_jf_j}$ are always monitored by $a_i$ and $f_j$. This means that only the arcs of the form $\overrightarrow{e_ig_i}$ are not monitored by the vertices of $A$.

Suppose we have a solution set $C$ for $\Phi$, that is $C = \{v_{i_1}, \dots, v_{i_k}\}$ such that for each $i\in \{1, \dots, m\}$, there exists $\alpha$ such that $e_i$ is incident to $v_{i_\alpha}$. We claim that $A\cup C$ is an MAG-set of $\overrightarrow{H(\Phi)}$, where we identify the elements of $C$ to their corresponding vertex in $\overrightarrow{H(\Phi)}$, and this set is of size $k + 2n +2m$. Recall that all the arcs of $\overrightarrow{H(\Phi)}$ except those of the form $\overrightarrow{e_jg_j}$ are monitored by the vertices of $A$. Let us consider $j\in \{1, \dots, n-1\}$. There exists $\alpha\in \{1, \dots, k\}$ such that $e_j$ is incident to $v_{i_\alpha}$ where $v_{i_\alpha}\in C$. There is a unique shortest path from $v_{i_\alpha}$ to $g_j$, and it contains $\overrightarrow{e_jg_j}$, thus monitoring it.
\end{proof}

\begin{lemma}
\label{lem:hard2}
    Let $\overrightarrow{H(\Phi)}$ be an oriented graph obtained from $G(\Phi)$ by Construction~\ref{cons:hardness} with $|V(H(\Phi))| = 4n + 3m$. If $\overrightarrow{H(\Phi)}$ has an MAG-set of size $k + 2n + 2m$, then $\Phi$ has a solution set of size $k$.
\end{lemma}
\begin{proof}
     Suppose $M$ is an MAG-set of $\overrightarrow{H(\Phi)}$ of size $k + 2n + 2m$. We shall build $M'$ a solution for $\Phi$ of size $k$. As noted before, $M$ must contain all the vertices of $A$, leaving exactly $k$ other vertices. Since $M$ is an MAG-set, for any $j\in \{1, \dots, n \}$, the arc $\overrightarrow{e_jg_j}$ is monitored. This implies that there exist some vertices $u, v\in V(H(\Phi))$ such that $u, v\in M$ and a shortest path from $u$ to $v$ akin to $u \ldots e_j g_j \ldots v$. Note that $u$ may be equal to $e_j$ and since $g_j$ is a sink, it must be that $v=g_j$. Moreover, $u$ cannot be a vertex of the type $a_i$ for any index $i$, because the arc $a_ig_j$ is always a shorter path, hence $u \notin A$. This implies that $u=e_j$ or $u\in \{v_i, c_i\}$ for some index $i$. If $u=e_j$, then we choose an $i$ such that $\overrightarrow{v_ie_j}$ is an arc of $\overrightarrow{H(\Phi)}$ and add $v_i$ to $M'$. Note that there can be multiple such $i$, but we fix an arbitrary one. If $u=c_i$, then we simply add the corresponding $v_i$ to $M'$. If $u=v_i$, we keep the same $v_i$. Either way, we add $v_i$ to $M'$. This way, $|M'| \leq |M| - |A|$, hence is of size $k$. Moreover, our choice ensures that for every $e_j$, there is a vertex $v_i \in M'$ such that $\overrightarrow{v_ie_j}$ is an arc of $\overrightarrow{H(\Phi)}$. That is, for every edge $e_j$ in $\Phi$, there is an incidence vertex $v_i$ in $M'$. Hence, $M'$ is a vertex cover of $\Phi$.
\end{proof}

\begin{proof}[Proof of Theorem~\ref{thm:hardness}]
    A certificate for the MAG-set on a graph $\overrightarrow{G}$ is a subset of $V(\overrightarrow{G})$. For every pair of vertices $u$ and $v$ in the certificate, whether they monitor an arc $\overrightarrow{a}$ can be checked in polynomial time by finding and comparing the shortest distances between $u$ and $v$ in $\overrightarrow{G}$ and in $\overrightarrow{G} - \overrightarrow{a}$. Hence, the problem lies in NP. Lemmas~\ref{lem:hard1} and~\ref{lem:hard2} prove that an instance $\Phi$ of \textsc{Vertex Cover} of size at most $n + m$, and a solution of size $k$, can be reduced to an instance of \textsc{$k$-MAG-set} of size $4n + 3m$ with a solution of size $k + 2n + 2m$. Hence, the problem is NP-complete. To complete the proof of the theorem, observe that if $\Phi$ is a \textsc{Vertex Cover} instance restricted to graphs with maximum degree $3$, then $\overrightarrow{H(\Phi)}$ is an instance of \textsc{$k$-MAG-Set} of maximum degree $4$, and that there are neither any oriented cycles in $\Phi$, nor does Construction~\ref{cons:hardness} create any oriented cycles in $\overrightarrow{H(\Phi)}$.
\end{proof}

\section{Conclusions}
\label{sec6}
In this work, we have defined and initiated the study of monitoring arc-geodetic sets of graphs. We conclude by mentioning several directions with scope for carrying out further research and explorations.
\paragraph{Operations on graphs:}
Exploring the effects of graph operations like vertex deletion, contraction of an arc, etc. on the monitoring arc-geodetic number of the graph, would be an interesting prospect. To construct the $k$th power $\overrightarrow{G}^k$ of a graph $\overrightarrow{G}$, we add an arc between vertices at a distance of at most $k$. Consequently, vertices in an MAG-set of $\overrightarrow{G}$ might be connected via a shorter path in $\overrightarrow{G}^k$ and hence monitor fewer arcs in $\overrightarrow{G}^k$, potentially increasing the monitoring arc-geodetic number. Investigating the existence of bounds on $mag(\overrightarrow{G})$ in terms of $mag(\overrightarrow{G}^k)$ is another intriguing area of research. 

We can also ask about the relation between the monitoring arc-geodetic number of graphs and their Cartesian product, categorical product, tensor product, etc. Such studies hold promise due to the existence of works like~\cite{foucaud2024monitoring,HASLEGRAVE202379} which have addressed and studied analogous questions on monitoring edge-geodetic sets.

\paragraph{Other related parameters:} Geodetic sets, edge-geodetic sets, strong edge-geodetic sets, distance edge monitoring sets etc., are distance distinguishing parameters that are related to the monitoring edge-geodetic sets~\cite{foucaud2024monitoring}. While research on geodetic sets for oriented graphs is well-established~\cite{chartrand2000geodetic,lu2007geodetic,chang2004geodetic,dong2009upper,hung2006hull,hung2009hull,araujo2022hull,farrugia2005orientable}, the definition of the oriented analogues of the other parameters and the study of their properties is an open area of research. In particular, investigating the inter-relationships of the oriented versions of these parameters, especially with the monitoring arc-geodetic number, is of significant interest. 

\paragraph{Graphs with specific properties:} Construction~\ref{cons:mag<mag-} gave us a family of graphs where the monitoring edge-geodetic number and the lower monitoring arc-geodetic number differed by one. This prompts the question, can we find examples of families of graphs where this difference can be arbitrarily large? To be precise, we can ask the following question:

\begin{question}
For any positive integer $j$, $j \geq 3$, can we find a connected graph $G_j$ such that $mag^-(G_j) - meg(G_j) \geq j$?
\end{question}

We can also explore the following intriguing question about the existence of graphs with specific monitoring arc-geodetic numbers.

\begin{question}
    For positive integers $n,m$ and $k$ with $k \geq 2$ and $m \leq n(n - 1)/2$, does there exist an oriented graph $\overrightarrow{G}$ such that $|V(\overrightarrow{G})| = n$, $|A(\overrightarrow{G})| = m$ and $mag(\overrightarrow{G}) = k$?
\end{question}

Determining the values of $n, m$ and $k$ for which the above question has an answer is non-trivial. For instance, our study on tournaments has revealed that if $m = n(n - 1)/2$, then $k$ must be either $n - 1$ or $n$; the graph cannot exist for any other values of $k$ under these conditions.

This question can be addressed for some specific values of $n, m$ and $k$ by considering a balanced binary tree on $n$ vertices and arcs oriented from the root to the leaves. The monitoring arc-geodetic number of this graph is one more than the number of leaves. By adjusting the orientations of some arcs as desired, oriented graphs with different values of $m$ and $k$ are obtained for fixed $n$. Studying the bounds of $m$ and $k$ for which such a construction is feasible could also be a viable direction of research.

The following question can be considered as an intermediate value problem for monitoring arc-geodetic numbers of graphs.
\begin{question}
    For every connected graph $G$ and positive integer $k$ with $mag^-(G) \leq k \leq mag^+(G)$, does there exist an orientation $\overrightarrow{G}$ of $G$ such that $mag(\overrightarrow{G}) = k$?
\end{question}
In fact, the answer to this question is in the negative since Proposition~\ref{prop:magcycle} proves that for $n \geq 3, S(C_n) = \{3\} \cup \{2k \colon 2 \leq 2k \leq n\}$. But it is interesting to study whether cycles are the only such exceptional family. Answering such questions will provide valuable insights into properties of monitoring arc-geodetic sets.

\paragraph{Relation between $mag^+$ and $mag^-$:} For an undirected graph $G$, the orientations with monitoring arc-geodetic number equal to $mag^-(G)$ and $mag^+(G)$, could be considered, in some sense, to be the best and worst orientations respectively, of $G$. Therefore, investigating graphs with high $mag^-$ and low $mag^+$ values is of significant interest. For instance, the only graphs $G$ with $mag^-(G) = 2$ are paths and cycles. But what can we say about the graphs with $mag^-(G) = 3$?

\begin{question}
    Characterize the undirected graphs $G$ with $mag^-(G) = 3$.
\end{question}

To solve this question, one could draw inspiration from~\cite{chartrand2000geodetic} where a similar study has been carried out for lower geodetic numbers.

MAG-extremal graphs and the \textsc{MAG$^+$-set} problem study the case when $mag^+(G) = n$. This prompts the following question:

\begin{question}
    Characterize the undirected graphs $G$ with the property $mag^-(G) = |V(G)| - 1$.
\end{question}
A related question which could help in the study of this problem is the following.

\begin{question}
    For an oriented graph $\overrightarrow{G}$ with $|V(\overrightarrow{G})| \geq 3$, what are the properties of a vertex $v \in V(\overrightarrow{G})$ such that $v$ is never a part of any minimal MAG-set of $\overrightarrow{G}$?
\end{question}
Identifying the properties of such a vertex $v$ will help determine the graphs with $mag^+ = n - 1$.

We know by definition that trivial bounds exist for $mag^+$ and $mag^-$ of a graph, that is, 
\begin{equation}
\label{eq:conc}
    2 \leq mag^-(G) \leq mag^+(G) \leq n.
\end{equation}
Improving this inequality is a nice challenge. For a graph $G$ of order $n$, define $h(G) = mag^+(G) - mag^-(G)$. We know from Equation~\ref{eq:conc} that $0 \leq h(G) \leq n - 2$. This leads to the following questions.

\begin{question}
    Are there arbitrarily large graphs $G$ with the property that $mag^-(G) = mag^+(G)$?
\end{question}

The graph consisting of a single edge trivially satisfies this problem, but we have seen in this paper that several classes of graphs such as paths, cycles, tournaments, etc. do not. If one can show that for any graph $G$, we can find two orientations with different monitoring arc-geodetic numbers, then the above question will be answered in the negative. In fact, it is known~\cite{farrugia2005orientable} that there are no such graphs with respect to the upper and lower orientable geodetic numbers defined by Chartrand and Zhang~\cite{chartrand2000geodetic} and a similar answer can be expected for our question too.

\begin{question}
    What is the minimum value of $h(G)$ for an undirected graph $G$?
\end{question}

For example, for paths, we have seen that $h(G) = n - 2$. Studying the value of $h(G)$ across various graph classes, and identifying graphs for which the value of $h(G)$ is minimum, could provide valuable insights.

\paragraph{Algorithmic problems:} The optimisation version of the \textsc{$k$-MAG-set} problem involves finding the minimum MAG-set of $\overrightarrow{G}$. Studying this problem allows us to compare and contrast results with the undirected case~\cite{bilo2024inapproximability}. Additionally, exploring the Fixed Parameter Tractability (FPT) of this problem remains an uncharted area for future research in this field.

\bigskip
\noindent\textbf{Acknowledgements:} This research was financed by the IFCAM project ``Applications of graph homomorphisms'' (MA/IFCAM/18/39) and SERB-MATRICS ``Oriented chromatic and clique number of planar graphs'' (MTR/2021/000858). Florent Foucaud was financed by the French government IDEX-ISITE initiative CAP 20-25 (ANR-16-IDEX-0001), the International Research Center ``Innovation Transportation and Production Systems" of the I-SITE CAP 20-25, and the ANR project GRALMECO (ANR-21-CE48-0004). Pavan P D was financed by Academy of Finland grant number 338797.

\bibliographystyle{abbrv}
\bibliography{references(oriented)}

\begin{thebibliography}{10}

\bibitem{araujo2022hull}
J.~Araújo and P.~Arraes.
\newblock Hull and geodetic numbers for some classes of oriented graphs.
\newblock {\em Discrete Applied Mathematics}, 323:14--27, 2022.

\bibitem{bampas2015network}
E.~Bampas, D.~Bil{\`o}, G.~Drovandi, L.~Gual{\`a}, R.~Klasing, and G.~Proietti.
\newblock Network verification via routing table queries.
\newblock {\em Journal of Computer and System Sciences}, 81(1):234--248, 2015.

\bibitem{bejerano2003robust}
Y.~Bejerano and R.~Rastogi.
\newblock Robust monitoring of link delays and faults in ip networks.
\newblock In {\em IEEE INFOCOM 2003. Twenty-second Annual Joint Conference of the IEEE Computer and Communications Societies (IEEE Cat. No. 03CH37428)}, volume~1, pages 134--144. IEEE, 2003.

\bibitem{bilo2024inapproximability}
D.~Bil{\`{o}}, G.~Colli, L.~Forlizzi, and S.~Leucci.
\newblock On the inapproximability of finding minimum monitoring edge-geodetic sets (short paper).
\newblock In {\em Proceedings of the 25th Italian Conference on Theoretical Computer Science, Torino, Italy, September 11-13, 2024}, volume 3811 of {\em {CEUR} Workshop Proceedings}, pages 219--224. CEUR-WS.org, 2024.

\bibitem{chang2004geodetic}
G.~J. Chang, L.-D. Tong, and H.-T. Wang.
\newblock Geodetic spectra of graphs.
\newblock {\em European Journal of Combinatorics}, 25(3):383--391, 2004.

\bibitem{chartrand2000geodetic}
G.~Chartrand and P.~Zhang.
\newblock The geodetic number of an oriented graph.
\newblock {\em European Journal of combinatorics}, 21(2):181--189, 2000.

\bibitem{foucaud2023monitoring}
S.~R. Dev, S.~Dey, F.~Foucaud, N.~Krishna, and L.~R. Sulochana.
\newblock Monitoring edge-geodetic sets in graphs.
\newblock {\em arXiv preprint arXiv:2210.03774}, 2023.

\bibitem{dong2009upper}
L.~Dong, C.~Lu, and X.~Wang.
\newblock The upper and lower geodetic numbers of graphs.
\newblock {\em Ars Combin}, 91:401--409, 2009.

\bibitem{farrugia2005orientable}
A.~Farrugia.
\newblock Orientable convexity, geodetic and hull numbers in graphs.
\newblock {\em Discrete applied mathematics}, 148(3):256--262, 2005.

\bibitem{foucaud2022monitoring}
F.~Foucaud, S.-S. Kao, R.~Klasing, M.~Miller, and J.~Ryan.
\newblock Monitoring the edges of a graph using distances.
\newblock {\em Discrete Applied Mathematics}, 319:425--438, 2022.

\bibitem{foucaud2024monitoring}
F.~Foucaud, C.~Marcille, Z.~M. Myint, R.~Sandeep, S.~Sen, and S.~Taruni.
\newblock Monitoring edge-geodetic sets in graphs: extremal graphs, bounds, complexity.
\newblock In {\em Conference on Algorithms and Discrete Applied Mathematics}, pages 29--43. Springer, 2024.

\bibitem{foucaud2023monitoring2}
F.~Foucaud, K.~Narayanan, and L.~Ramasubramony~Sulochana.
\newblock Monitoring edge-geodetic sets in graphs.
\newblock In A.~Bagchi and R.~Muthu, editors, {\em Algorithms and Discrete Applied Mathematics}, pages 245--256, Cham, 2023. Springer International Publishing.

\bibitem{10.1145/800119.803884}
M.~R. Garey, D.~S. Johnson, and L.~Stockmeyer.
\newblock Some simplified {NP}-complete problems.
\newblock In {\em Proceedings of the Sixth Annual ACM Symposium on Theory of Computing}, STOC '74, page 47–63, New York, NY, USA, 1974. Association for Computing Machinery.

\bibitem{harary1993geodetic}
F.~Harary, E.~Loukakis, and C.~Tsouros.
\newblock The geodetic number of a graph.
\newblock {\em Mathematical and Computer Modelling}, 17(11):89--95, 1993.

\bibitem{HASLEGRAVE202379}
J.~Haslegrave.
\newblock Monitoring edge-geodetic sets: Hardness and graph products.
\newblock {\em Discrete Applied Mathematics}, 340:79--84, 2023.

\bibitem{hung2006hull}
J.-T. Hung.
\newblock {\em The Hull Numbers of Orientations of Graphs}.
\newblock PhD thesis, National Sun Yat-sen University, 2006.

\bibitem{hung2009hull}
J.-T. Hung, L.-D. Tong, and H.-T. Wang.
\newblock The hull and geodetic numbers of orientations of graphs.
\newblock {\em Discrete mathematics}, 309(8):2134--2139, 2009.

\bibitem{karp2010reducibility}
R.~M. Karp.
\newblock {\em Reducibility among combinatorial problems}.
\newblock Springer, 2010.

\bibitem{kim2004geodetic}
B.~K. Kim.
\newblock The geodetic number of a graph.
\newblock {\em Journal of Applied Mathematics and Computing}, 16(1/2):525--532, 2004.

\bibitem{lu2007geodetic}
C.-h. Lu.
\newblock The geodetic numbers of graphs and digraphs.
\newblock {\em Science in China Series A: Mathematics}, 50(8):1163--1172, 2007.

\bibitem{ma2024monitoring}
R.~Ma, Z.~Ji, Y.~Yao, and Y.~Lei.
\newblock Monitoring-edge-geodetic numbers of radix triangular mesh and sierpi{\'n}ski graphs.
\newblock {\em International Journal of Parallel, Emergent and Distributed Systems}, 39(3):353--361, 2024.

\bibitem{moore2011nature}
C.~Moore and S.~Mertens.
\newblock {\em The Nature of Computation}.
\newblock OUP Oxford, 2011.

\bibitem{tan2023monitoring}
A.~Tan, W.~Li, X.~Wang, and X.~Li.
\newblock Monitoring edge-geodetic numbers of convex polytopes and four networks.
\newblock {\em International Journal of Parallel, Emergent and Distributed Systems}, 38(4):301--312, 2023.

\bibitem{xu2024monitoring}
X.~Xu, C.~Yang, G.~Bao, A.~Zhang, and X.~Shao.
\newblock Monitoring-edge-geodetic sets in product networks.
\newblock {\em International Journal of Parallel, Emergent and Distributed Systems}, 39(2):264--277, 2024.

\end{thebibliography}

\end{document}